\documentclass[a4paper,11pt]{article}

\usepackage{graphicx}
\usepackage[T1]{fontenc}

\usepackage{amsmath,amsfonts,amssymb,amsthm,enumitem,bbm,xcolor,mathtools}
\usepackage[margin=2cm]{geometry}
\usepackage{etoolbox}
\usepackage{tikz}
\usepackage[backgroundcolor=green!40]{todonotes}
\usepackage{hyperref}
\usepackage{cleveref}

\usepackage{tabularx,environ}

\makeatletter
\newcommand{\problemtitle}[1]{\gdef\@problemtitle{#1}}
\newcommand{\probleminput}[1]{\gdef\@probleminput{#1}}
\newcommand{\problemquestion}[1]{\gdef\@problemquestion{#1}}
\NewEnviron{decisionproblem}{
  \problemtitle{}\probleminput{}\problemquestion{}
  \BODY
  \par\addvspace{.5\baselineskip}
  \noindent
  \begin{tabularx}{\textwidth}{l X c}
    \multicolumn{2}{l}{\DecisionProblem{\@problemtitle}} \\   
    \textbf{Input:}     & \@probleminput \\ 
    \textbf{Question:}  & \@problemquestion 
  \end{tabularx}
  \par\addvspace{.5\baselineskip}
}
\makeatother

\theoremstyle{plain}
\newtheorem{theorem}{Theorem}[section]

\newtheorem{corollary}  [theorem]{Corollary}

\newtheorem{example}    [theorem]{Example}

\newtheorem{lemma}      [theorem]{Lemma}

\newtheorem{proposition}[theorem]{Proposition}

\theoremstyle{definition}
\newtheorem{claim}  {Claim}
\newtheorem*{claim*}{Claim}

\newcommand{\Z} {\mathbb{Z}}

\newcommand{\Define}            [1] {\textbf{#1}}
\newcommand{\DecisionProblem}   [1] {\textnormal{\textsc{#1}}}
\newcommand{\ComplexityClass}   [1] {\textnormal{\textbf{#1}}}

\renewcommand{\P}     {\ComplexityClass{P}}
\newcommand{\NP}    {\ComplexityClass{NP}}
\newcommand{\coNP}  {\ComplexityClass{coNP}}

\newcommand{\zero}      { \mathbf{0} }
\newcommand{\one}       { \mathbf{1} }
\newcommand{\meet}      { \wedge }
\newcommand{\join}      { \vee }
\newcommand{\bigmeet}   { \bigwedge }
\newcommand{\bigjoin}   { \bigvee }

\newcommand{\CharacteristicVector}   { \chi }

\newcommand{\Network}       { \mathsf{F} }
\newcommand{\Dominating}    { \mathsf{D} }
\newcommand{\Independent}   { \mathsf{I} }
\newcommand{\Kernel}        { \mathsf{K} }
\newcommand{\MIS}           { \mathsf{M} }

\newcommand{\DominatingSets}        { \mathrm{D} }
\newcommand{\IndependentSets}       { \mathrm{I} }
\newcommand{\Kernels}               { \mathrm{K} }
\newcommand{\MaximalIndependentSets}{ \mathrm{M} }

\newcommand{\Fix}{\mathrm{Fix}}

\newcommand{\Neighbourhood}     { N }
\newcommand{\InNeighbourhood}   { \Neighbourhood^{ \mathrm{in} } }
\newcommand{\OutNeighbourhood}  { \Neighbourhood^{ \mathrm{out} } }

\newcommand{\Twins}             [1] {\langle #1 \rangle}
\newcommand{\Benjamins}         {\mathrm{B}}

\newcommand{\Geodesic}      { \xmapsto{ \mathrm{geo} } }

\newcommand{\Variable}{\alpha}
\newcommand{\Variables}{A}
\newcommand{\Literal}{\beta}
\newcommand{\Literals}{B}
\newcommand{\Clause}{\gamma}
\newcommand{\Clauses}{\Gamma}
\newcommand{\OtherVariable}{\omega}
\newcommand{\OtherVariables}{\Omega}

\title{Generalising the maximum independent set algorithm via Boolean networks}

\author{Maximilien Gadouleau and David C. Kutner\footnote{Department of Computer Science, Durham University, Durham, UK. \url{ {m.r.gadouleau, david.c.kutner}@durham.ac.uk } } }

\begin{document}

\maketitle

\begin{abstract}
A simple greedy algorithm to find a maximal independent set (MIS) in a graph starts with the empty set and visits every vertex, adding it to the set if and only if none of its neighbours are already in the set. In this paper, we consider the generalisation of this so-called MIS algorithm by allowing it to start with any set of vertices and we prove the hardness of many decision problems related to this generalisation. Our results are based on two main strategies. Firstly, we view the MIS algorithm as a sequential update of a Boolean network, which we shall refer to as the MIS network, according to a permutation of the vertex set. The set of fixed points of the MIS network corresponds to the set of MIS of the graph. Our generalisation then consists in allowing to start from any configuration and to follow a sequential update given by a word of vertices. Secondly, we introduce the concept of a constituency of a graph, that is a set of vertices that is dominated by an independent set. Deciding whether a set of vertices is a constituency is \NP-complete; decision problems related to the MIS algorithm will be reduced from the \DecisionProblem{Constituency} problem or its variants. In this paper, we first consider universal configurations, i.e. those that can reach all maximal independent sets; deciding whether a configuration is universal is \coNP-complete. Second, we consider so-called fixing words, that allow to reach a MIS regardless of the starting configuration, and fixing permutations, which we call permises; deciding whether a permutation is fixing is \coNP-complete. Third, we consider permissible graphs, i.e. those graphs that have a permis. We construct large classes of permissible and non-permissible graphs, notably near-comparability graphs which may be of independent interest; deciding whether a graph is permissible is \coNP-hard. Finally, we generalise the MIS algorithm to digraphs. In this case, the algorithm uses the so-called kernel network, whose fixed points are the kernels of the digraph. Deciding whether the kernel network of a given digraph is fixable is \coNP-hard, even for digraphs that have a kernel. As an alternative, we introduce two further Boolean networks, namely the independent and the dominating networks, whose sets of fixed points contain all kernels. Unlike the kernel network, those networks are always fixable and their fixing word problem is in \P. 

\end{abstract}

\section{Introduction} \label{section:introduction}

\paragraph{The MIS algorithm} \label{paragraph:intro_mis_algorithm}

A simple greedy algorithm to find a maximal independent set (MIS) in a graph starts with the empty set and visits every vertex, adding it to the set if and only if none of its neighbours are already in the set. We shall refer to it as the MIS algorithm. Because the MIS algorithm always terminates in a maximal independent set, it has been the subject of a stream of work (see \cite{BFS12} and references therein). 

A core feature of the classical MIS algorithm is that the starting set of vertices is the empty set. However, the seminal observation of this paper is that this constraint can be lifted. Indeed, starting from any set of vertices and visiting each vertex once, removing a vertex if one of its neighbours already appears in the set, one always terminates at an independent set. Moreover, starting from any independent set and visiting each vertex once, one always terminates at a MIS. Thus, iterating over the vertex set twice is sufficient to obtain a MIS from any starting set of vertices.

As such, the scope of this paper is the generalisation of the MIS algorithm, where one can start with any (not necessarily independent) set of vertices, and one can visit vertices in any order with possible repetitions. Of course, some sequences of vertices will guarantee reaching a MIS from any starting configuration – we shall call those fixing words, while others will not, as in Example \ref{example:P3} below.

\begin{example} \label{example:P3}

Consider the path on three vertices: 
\begin{center}
    \begin{tikzpicture}[xscale=4]

\begin{scope}[yshift=0cm]    
    \node[draw, circle] (a) at (0,0) {$a$};
    \node[draw, circle] (b) at (1,0) {$b$};
    \node[draw, circle] (c) at (2,0) {$c$};

    \draw (a) -- (b) -- (c);
\end{scope}
\end{tikzpicture}
\end{center}

\noindent This graph has two MIS, namely $\{ a,c \}$ and $\{ b \}$. Starting at the empty set, the MIS algorithm terminates at $\{ b \}$ if the sequence begins with $b$ (i.e. $w = bac$ or $w = bca$) or at $\{a, c\}$ otherwise (i.e. $w \in \{abc, acb, cab, cba\}$). 

For this graph, $abc$ is not a fixing word: if one starts from the set $\{ b, c \}$, then one terminates at $\{ c \}$. However, $acb$ is a fixing word: if the starting set contains $b$, then one terminates at $\{ b \}$; otherwise one terminates at $\{ a, c \}$.

Finally, for this graph, the words $w=abcabc$ and $w=acbacb$ are fixing words, as are any words of the form $w=w^1w^2$, where $w^1$ and $w^2$ are permutations of the vertex set.

\end{example}


\paragraph{Contributions for graphs}

Below we give a summary of our contributions for the MIS algorithm on graphs.

When starting from the empty set, the MIS algorithm is able to reach any possible maximal independent set (if the algorithm goes through the MIS first, then it would terminate with that MIS). However, the empty set is not the only set with that property: the full set of vertices also allows that (this time, if the algorithm finishes with a MIS). In Theorem \ref{theorem:complexity_universal_configuration_K}, we prove that deciding whether a set of vertices can reach every MIS is \coNP-complete.

As we showed in Example \ref{example:P3}, though iterating over the whole set of vertices twice is always sufficient to reach a MIS, it is not always necessary. Consequently, we ask: what are the sequences of vertices which always reach a MIS, regardless of the starting set of vertices? In Theorem \ref{theorem:complexity_fixing_word_K}, we prove that deciding whether a sequence offers that guarantee is \coNP-complete.

Since the MIS algorithm visits each vertex exactly once, we also consider permutations of vertices that are guaranteed to reach a MIS; we call those permises. In Theorem \ref{theorem:complexity_permis}, we prove that deciding whether a permutation of vertices is a permis is \coNP-complete. A graph that admits a permis is called permissible. Not all graphs are permissible; the smallest non-permissible graph is the heptagon. We exhibit large classes of permissible and non-permissible graphs. In particular, we introduce near-comparability graphs and classify them in Theorem \ref{theorem:near-comparability}; they naturally generalise comparability graphs and can be recognised in polynomial time. We prove that near-comparability graphs are permissible in Proposition \ref{proposition:near-comparability_permissible}. In Theorem \ref{theorem:complexity_permissible}, we prove that deciding whether a graph is permissible is \coNP-hard. There is no obvious candidate for a no-certificate, so it may be that the problem is not actually in \coNP.

In some situations, one can skip some vertices and still guarantee a MIS is reached. For instance, in the complete graph, one can simply update all but one vertex and still reach a maximal independent set, from any starting configuration. We prove in Theorem \ref{theorem:complexity_fixing_set} that deciding whether a given set of vertices can be skipped is \coNP-complete. We also prove in Theorem \ref{theorem:complexity_non-trivial_fixing_set} that deciding whether \emph{any} vertices can be skipped is \coNP-complete. 

\paragraph{Boolean networks} \label{paragraph:intro_BN}

Our main tool is Boolean networks. A configuration on a graph $G = (V,E)$ is $x \in \{0,1\}^V$, i.e. the assignment of a Boolean state to every vertex of the graph. A Boolean network is a mapping $\Network : \{ 0, 1 \}^V \to \{ 0, 1 \}^V$ that acts on the set of configurations. Boolean networks are used to model networks of interacting entities. As such, it is natural to consider a scenario wherein the different entities update their state at different times. This gives rise to the notion of sequential (or asynchronous) updates, by updating the state of one vertex at a time; a word $w$ then gives the order in which those vertices are updated (with repeats allowed in general). Since the original works by Kauffman \cite{Kau69} and Thomas \cite{Tho73}, asynchronous updates have been widely studied, both in terms of modelling purposes and of dynamical analysis (see \cite{Bor08, ARS23} and references therein). The problem of whether a Boolean network converges (sequentially) goes back to the seminal result by Robert on acyclic interaction graphs \cite{Rob80}; further results include  \cite{Gol85,GM12,NS17}.
Recently, \cite{AGRS20} introduced the concept of a fixing word: a word $w$ such that updating vertices according to $w$ will always lead to a fixed point, regardless of the initial configuration. Fixing words are a natural feature of Boolean networks, for two main reasons. Firstly, almost all networks with a fixed point, and hence a positive asymptotic proportion of all networks, have fixing words \cite{BGS93}. Secondly, large classes of Boolean networks, including monotone networks and networks with acyclic interaction graphs, have short fixing words (of length at most cubic in $|V|$) \cite{AGRS20,GR18}.

We refer to the Boolean network where the update function is the conjunction of all the negated variables in the neighbourhood of a vertex as the MIS network on the graph, i.e. $\MIS : \{ 0, 1 \}^V \to \{ 0, 1 \}^V$ with $\MIS( x )_v = \bigmeet_{u \sim v} \neg x_u$ for all $v \in V$. The MIS network was highlighted in \cite{RR13,ARS14}, where the fixed points of different conjunctive networks on (directed) graphs are studied. In particular, \cite{RR13} shows that the set of fixed points of the MIS network is the set of (configurations whose supports are) maximal independent sets of the graph. It is further shown in \cite{ARS14} that for square-free graphs, the MIS network is the conjunctive network that maximises the number of fixed points. 

The MIS algorithm can be interpreted in terms of Boolean networks as follows: starting with the all-zero configuration $x$, update one vertex $v$ at a time according to the update function $\MIS( x )_v = \bigmeet_{u \sim v} \neg x_u$. Once all vertices have been updated, we obtain the final configuration $y$ where the set of ones is a maximal independent set, regardless of the order in which the vertices have been updated. As such, fixing words of the MIS network correspond to sequences of vertices that guarantee that the MIS algorithm terminates for any starting set of vertices. The seminal observation of this paper is that for any permutation $w$, the word $ww$ is a short fixing word (of length $2|V|$).

\paragraph{Self-stabilization and distributed computing} \label{paragraph:distributed}
Some of our results may be applied in the context of distributed computing. Similar to Boolean networks, distributed algorithms produce an output through local, independent updates of nodes in a fixed topology.
Asynchronous models for distributed computing do not assume a bound on message delay \cite{AABHBB11}, making them less relevant here. Consequently, we focus on synchronous models, in which time is discrete and all nodes perform a SEND-RECEIVE-UPDATE loop synchronously at each time step. The algorithm executed at all nodes is identical, and the state of some node at time $t$ depends only on the state of (all nodes in) its inclusive neighbourhood at time $t-1$. Some key differences with the Boolean Network setting are worth emphasizing. For example, in standard models for synchronous distributed computation: nodes each have an (unbounded-size) internal state, and may solve arbitrarily hard problems during their UPDATE; messages sent may differ from the sender's state; nodes may choose not to SEND anything at all; and updates occur synchronously.

The problem of finding a MIS has been a focus of much study in this setting, including in the LOCAL \cite{Gha22}, CONGEST \cite{Gha19} and Beeping models \cite{AABCHK12,BBDK18}. LOCAL is characterized by its unrestricted message size, whereas CONGEST limits messages to $O(\log |V|)$ bits per outgoing edge. The Beeping model is a significant restriction, in which nodes can communicate only via beeps (which are indistinguishable) and silence \cite{CD19}. We refer the interested reader to \cite{CMRZ19} for a more complete treatment of this model's variants, which also includes a discussion of the distributed MIS problem in Sections 4.5 and 6.2.

An algorithm or procedure is said to be ``self-stabilizing'' if it is guaranteed to reach a legitimate state regardless of its initial state, and additionally will never reach an illegitimate state from a legitimate state \cite{Dij74, Sch93}. This notion has been integrated into the design of distributed algorithms \cite{LSW09} and is explicitly identified as a feature of the Beeping MIS algorithm given in \cite{AABCHK12}.

Our results do not directly apply to these models; in particular, we assume (and sometimes exploit) asynchronous and instantaneous updates. That is, each vertex's local update is immediately ``visible'' to all its neighbours. This differs from standard models of distributed computing, which generally incorporate some transmission delay (which is typically one unit in synchronous models, and controlled by an adversary in asynchronous models).

To emulate the MIS network $\MIS$ studied in the present work, it would then be sufficient for a distributed model to support: asynchronous updates (scheduled by an adversary or a helper) and instantaneous transmission.
We call the minimum length of time within which each node updates at least once a \emph{phase}. In the adversarial setting, our seminal observation translates to the fact that this protocol necessarily reaches a MIS within two phases from any starting configuration and self-stabilizes. In the helpful setting, a Permis is an update schedule which guarantees self-stabilization within a single phase. By Theorem \ref{theorem:complexity_fixing_word_K}, it is \coNP-complete to determine whether some update schedule satisfies this property; by Theorem \ref{theorem:complexity_permis}, the problem remains \coNP-complete even if the schedule is guaranteed to contain every node exactly once; and by Theorem \ref{theorem:complexity_permissible} it is \coNP-hard to determine whether any such update schedule exists at all for the given network. 
If the helpful scheduler is limited to some subset of nodes, Theorem \ref{theorem:complexity_fixing_set} means it is \coNP-hard to determine whether there is a self-stabilizing schedule which uses only that subset of nodes. Furthermore, Theorem \ref{theorem:complexity_non-trivial_fixing_set} entails that deciding whether there exists any such schedule using $n-1$ nodes (even allowing repetitions) is \coNP-hard.

\paragraph{Constituencies} \label{paragraph:constituencies}

The main tool for the hardness results in this paper is that of a constituency. A constituency is a set of vertices of a graph that is dominated by an independent set, i.e. $S$ is a constituency if there exists an independent set $I$ such that $S \subseteq \Neighbourhood(I)$. We believe that the constituency problem is of independent interest for a couple of reasons. Firstly, this is a natural definition for a set of vertices, but to the best of the authors' knowledge, it has not been considered in the literature yet. Secondly, the \DecisionProblem{Constituency} problem asks whether a set $S$ is a constituency. Unlike problems like \DecisionProblem{Clique}, \DecisionProblem{Independent Set} or \DecisionProblem{Vertex Cover}, the \DecisionProblem{Constituency} problem does not rely on an integer parameter. Nonetheless, \DecisionProblem{Constituency} is \NP-complete, while the problem of deciding whether a set is a clique (or independent set, or vertex cover) is clearly in \P. As such, \DecisionProblem{Constituency} provides a natural intractable graph problem whose input does not include an integer. We heavily use \DecisionProblem{Constituency} and its variants in our hardness proofs, and we believe that this problem could be used more broadly for reductions in the wider graph theory community.

\paragraph{Extension to digraphs} \label{paragraph:digraphs}

In this paper, we further generalise the MIS algorithm by applying it to the digraph case. In this case, a vertex is added to the set if and only if none of its in-neighbours are already in the set. The expected outcome is a kernel, i.e. a dominating independent set (equivalent to a maximal independent set if the digraph is a graph). Unfortunately, not all digraphs have a kernel: odd directed cycles provide an intuitive example of kernel-less digraphs. In fact, deciding whether a digraph has a kernel is \NP-complete (see \cite{BG09a}, p. 119). 

Nonetheless, the algorithm for digraphs now corresponds to sequential updates of the kernel network, with $\Kernel( x )_v = \bigmeet_{u \to v} \neg x_u$. Again, the set of fixed points of the kernel network is the set of (configurations whose supports are) kernels \cite{RR13}. As a side note, the kernel network has been heavily used in logic and philosophy. Indeed, Yablo discovered the first non-self-referential paradox in \cite{Yab93}. The construction for this paradox implicitly applies the fact that the kernel network on a transitive tournament on $\mathbb{N}$ has no fixed point. The study of acyclic digraphs that admit a paradox is continued further in \cite{RRM13}, where the kernel network is referred to as an $\mathcal{F}$-system. 

Some digraphs have a kernel and yet the corresponding kernel network does not have a fixing word. More strongly, in Theorem \ref{theorem:complexity_fixable} we show that deciding whether the kernel network has a fixing word is \coNP-hard, and we show in Theorem \ref{theorem:complexity_fixable_refined} that the result holds even for a restricted class of oriented digraphs that have a kernel.

We then consider two other Boolean networks, that are fixable for any digraph. Firstly, the independent network is given by $\Independent( x )_v = x_v \meet \bigmeet_{ u \to v } \neg x_u$; its set of fixed points consists of the independent sets of the digraph. We classify the fixing words of the independent network in Proposition \ref{proposition:fixing_words_I} and in Corollary \ref{corollary:complexity_fixing_word_I} we prove that deciding whether a word fixes the independent network is in \P. Secondly, the dominating network is given by $\Dominating( x )_v = x_v \join \bigmeet_{ u \to v } \neg x_u$; its set of fixed points consists of the dominating sets of the digraph. Similarly, we classify the fixing words of the dominating network in Proposition \ref{proposition:fixing_words_D} and in Corollary \ref{corollary:complexity_fixing_word_D} we prove that deciding whether a word fixes the dominating network is in \P.

\paragraph{Outline} \label{paragraph:outline}

The rest of the paper is organised as follows. Some necessary background is given in Section \ref{section:preliminaries}. Constituencies and districts are introduced in Section \ref{section:constituencies}, where some decision problems based on those are proved to be \NP- or \coNP-complete. The configurations that allow to reach any possible maximal independent set are determined in Section \ref{section:reachability_K}. Fixing words, fixing sets, and permises for the MIS network are studied in Section \ref{section:fixing_words_K}. Classes of permissible and non-permissible graphs are given in Section \ref{section:permissible_non-permissible}. The extension to digraphs is carried out in Section \ref{section:digraphs}, where we first consider fixing words of the kernel network and then study the independent and dominating networks instead. Finally, some conclusions and possible avenues for future work are given in Section \ref{section:conclusion}.

\section{Preliminaries} \label{section:preliminaries}

\subsection{Graphs and digraphs} \label{subsection:background_graphs}

Most of our contributions (Sections \ref{section:constituencies} to \ref{section:permissible_non-permissible}) will focus on (undirected) graphs. However, when we extend our focus to directed graphs, we shall view graphs as natural special cases of digraphs. As such, we give the background on graphs and digraphs in its full generality, i.e. for digraphs first, and then we make some notes about the special case of graphs.

By digraph, we mean an irreflexive directed graph, i.e. $G = (V, E)$ where $E \subseteq V^2 \setminus \{ (v,v) : v \in V \}$. We use the notation $u \to v$ to mean that $(u,v) \in E$. We say an edge $(u,v) \in E$ is \Define{symmetric} if $(v,u)$ is also an edge, and \Define{oriented} otherwise. We will sometimes emphasize that $(u,v)$ is symmetric by writing it $uv$ instead. 
For a vertex $v$, the open in-neighbourhood, closed in-neighbourhood, open out-neighbourhood and closed out-neighbourhood of the vertex $v$ are respectively defined as
\begin{alignat*}{3}
    \InNeighbourhood(v)   &= \{ u \in V : u \to v \}, &\quad
    \InNeighbourhood[v]   &= \InNeighbourhood(v) \cup \{ v \}, \\
    \OutNeighbourhood(v)  &= \{ u \in V : v \to u \}, &\quad
    \OutNeighbourhood[v]  &= \OutNeighbourhood(v) \cup \{ v \}.
\end{alignat*}
All of those are generalised to sets of vertices, e.g. $\InNeighbourhood(S) = \bigcup_{s \in S} \InNeighbourhood(s)$. Clearly, all notations above should reflect the dependence on the digraph $G$, e.g. $\InNeighbourhood(v ; G)$; we shall omit that dependence on any notation when the digraph is clear from the context.

For a digraph $G=(V,E)$ and set of vertices $S\subseteq V$, we call the digraph $(S, \{(u,v): (u,v)\in E \land \{u,v\} \subseteq S\})$ the \Define{induced subgraph} on $S$, which we denote $G[S]$. We denote $G-S$ the digraph $G[V\setminus S]$. A \Define{path} is a sequence of edges $v_1 \to v_2 \to \dots \to v_k$ where all vertices are distinct; a \Define{cycle} in a digraph is a sequence of edges $v_1 \to v_2 \to \dots \to v_k \to v_1$ where  only the first and the last vertices are equal. A digraph is \Define{strong} if for all vertices $u$ and $v$, there is a path from $u$ to $v$. A \Define{strong component} of $G$ is a subset of vertices $S$ such that $G[S]$ is strong, but $G[T]$ is not strong for all $T \supsetneq S$. A digraph is \Define{acyclic} if it has no cycles. An acyclic digraph has a \Define{topological order}, whereby $u \to v$ only if $u \le v$. For instance, the digraph where each vertex is a strong component of $G$ and $C \to C'$ if and only if $u \to u'$ for some $u \in C$, $u' \in C'$ is acyclic. If $C \to C'$ in that digraph, we say that $C$ is a \Define{parent component} of $C$; a strong component without any parent is called an \Define{initial component}.

We say that two vertices $u$ and $v$ are \Define{closed twins} if $\InNeighbourhood[ u ] = \InNeighbourhood[ v ]$. Accordingly, we say that the vertex $m$ is a \Define{benjamin} of $G$ if there is no vertex $v$ with $\InNeighbourhood[ v ] \subset \InNeighbourhood[ m ]$. We denote the set of benjamins of $G$ by $\Benjamins(G)$ and the corresponding induced subgraph by $G_\Benjamins = G[ \Benjamins(G) ]$. We say that a set of vertices $S$ is \Define{tethered} if there is an  edge $st$ between any $s \in S$ and any $t \in T = \Neighbourhood(S) \setminus S$.

An \Define{out-tree} is a digraph where all edges are oriented, one vertex has no in-neighbours (the so-called root), and all other vertices have exactly one in-neighbour each. A \Define{spanning out-forest} of a digraph $G$ rooted at a set $S \subseteq V$ is a collection of out-trees, each rooted at a different vertex of $S$, such that each out-tree is a subgraph of $G$ and each vertex of $G$ appears in exactly one out-tree. We shall use the following simple fact about spanning out-forests.

\begin{lemma} \label{lemma:out-rooted_forest}
If $G$ is strong, then for any nonempty $S \subseteq V$, $G$ has a spanning out-forest rooted at $S$.
\end{lemma}

\begin{proof}
Let $S = \{s_1, \dots, s_k\}$. For any $u \in V$, let $s_{n(u)}$ denote the nearest vertex in $S$ from $u$ with the smallest index. More formally, let $d(a,b)$ denote the length of a shortest path from $a$ to $b$, then we define $s_{n(u)}$ such that $d( s_{n(u)}, u ) < d( s_l, u )$ for all $1 \le l \le n(u)$ and $d( s_{n(u)}, u ) \le d(s_m, u)$ for all $n(u) \le m \le k$.
For all $1 \le l \le k$, let $T_l = \{ u \in V : n(u) = l \}$. It is clear that if $v$ is on a shortest path from $s_{n(u)}$ to $u$, then $n(v) = n(u)$ (in particular, $n(s_l) = l$). Therefore, each $T_l$ has a spanning out-forest rooted at $s_l$. The union of all the $T_l$ trees forms the desired spanning out-forest rooted at $S$.
\end{proof}

A digraph is \Define{undirected} if all its edges are symmetric; which we shall simply call a \Define{graph}. We then denote $\Neighbourhood(v) = \InNeighbourhood(v) = \OutNeighbourhood(v)$, which we call the \Define{neighbourhood} of $v$. A strong graph is called \Define{connected}, and the (initial) strong components of a graph are called its \Define{connected components}. In a graph, if $u \to v$, then $v \to u$, which we shall denote by $u \sim v$. Lemma \ref{lemma:out-rooted_forest} applied to graphs is given below -- a \Define{spanning forest} of a graph $G$ rooted at a set $S \subseteq V$ is a collection of trees, each rooted at a different vertex of $S$, such that each tree is a subgraph of $G$ and each vertex of $G$ appears in exactly one tree.

\begin{corollary} \label{corollary:rooted_forest}
If $G$ is connected, then for any nonempty $S \subseteq V$, $G$ has a spanning forest rooted at $S$.
\end{corollary}

\subsection{Boolean networks} \label{subsection:background_BN}

A \Define{configuration} on a digraph $G = (V,E)$ is $x \in \{0,1\}^V = (x_v : v \in V)$, where $x_v \in \{0,1\}$ is the state of the vertex $v$ for all $v$. We denote $\one( x ) = \{ v \in V: x_v = 1 \}$ and $\zero( x ) = \{ v \in V: x_v = 0 \}$. Conversely, for any set of vertices $S \subseteq V$, the \Define{characteristic vector} of $S$ is the configuration $x = \CharacteristicVector(S)$ such that $\one( x ) = S$. For any set of vertices $S \subseteq V$, we denote $x_S = (x_v: v \in S)$. We denote the all-zero (all-one, respectively) configuration by $0$ (by $1$, respectively), regardless of its length. 

We consider the following kinds of sets of vertices of, and accordingly configurations on, a digraph $G$:
\begin{enumerate}
    \item \label{item:independent_set}
    An \Define{independent set} $I$ is a set such that $(i, j) \notin E$ for all $i,j \in I$. (In other words, $\OutNeighbourhood(I) \cap I = \emptyset$.) Every digraph $G$ has an independent set, namely the empty set $\emptyset$. The collection of characteristic vectors of independent sets of $G$ is denoted by $\IndependentSets(G)$. 

    \item \label{item:dominating_set}
    A \Define{dominating set} $D$ is a set such that for every vertex $v \in V$, either $v \in D$ or there exists $u \in D$ such that $(u,v) \in E$. (In other words, $\OutNeighbourhood(D) \cup D = V$.) Every digraph $G$ has an dominating set, namely $V$. The collection of characteristic vectors of dominating sets of $G$ is denoted by $\DominatingSets(G)$.

    \item \label{item:kernel}
    A \Define{kernel} $K$ is a dominating independent set. (In other words, $\OutNeighbourhood(K) \uplus K = V$.) Not all digraphs have a kernel, for instance the directed cycle $\vec{C}_n$ (with vertex set $\Z_n$ and edges $(v, v+1)$ for all $v \in \Z_n$) does not have a kernel whenever $n \ge 3$ is odd. The collection of characteristic vectors of kernels of $G$ is denoted by $\Kernels(G)$. 

    \item \label{item:maximal_independent_set}
    If $G$ is a graph, then a kernel is a \Define{maximal independent set} of $G$, i.e. an independent set $I$ such that there is no independent set $J \supset I$. Every graph has a maximal independent set. In order to highlight this special case of particular importance to this paper, the collection of characteristic vectors of maximal independent sets of $G$ is denoted by $\MaximalIndependentSets(G)$. 
\end{enumerate}

Let $w = w_1 \dots w_l \in V^*$ be a sequence of vertices, or briefly a \Define{word}. For any $a,b \in \{1, \dots, l \}$, we denote $w_{a:b} = w_a \dots w_b$ if $a \le b$ and $w_{a:b}$ is the empty sequence if $a > b$. We also denote by $[w] = \{ u \in V : \exists j \ w_j = u \}$ the set of vertices that $w$ visits. For any $S \subseteq V$, the subword of $w$ induced by $S$, denoted by $w[S]$, is obtained by deleting all the entries in $w$ that do not belong to $S$; alternatively, it is the longest subword of $w$ such that $[ w[S] ] \subseteq S$. A \Define{permutation} of $V$ is a word $w = w_{1 : n}$ such that $[w] = V$ and $w_a \ne w_b$ for all $a \ne b$.

A \Define{Boolean network} is a mapping $\Network : \{0,1\}^V \to \{0,1\}^V$. For any Boolean network $\Network$ and any $v \in V$, the update of the state of vertex $v$ is represented by the network $\Network^v: \{0,1\}^V \to \{0,1\}^V$ where $\Network^v(x)_v = \Network( x )_v$ and $\Network^v( x )_u = x_u$ for all other vertices $u$. We extend this notation to words as follows: if $w = w_1 \dots w_l$ then
\[
	\Network^w = \Network^{ w_l } \circ \dots \circ \Network^{ w_2 } \circ \Network^{ w_1 }.
\]
Unless otherwise specified, we let $x$ be the initial configuration, $w = w_1 \dots w_l$ be a word, $y = \Network^w( x )$ be the final configuration, and for all $0 \le a \le l$, $y^a = \Network^{ w_{1:a} }( x )$ be an intermediate configuration, so that $x = y^0$ and $y = y^l$. 

If there is a word $w$ such that $y = \Network^w( x )$, we say that $y$ is \Define{reachable} from $x$, and we write $x \mapsto_\Network y$. For any two configurations $x$ and $y$, we denote $\Delta( x, y ) = \{ v \in V : x_v \ne y_v \}$. An \Define{$\Network$-geodesic} from $x$ to $y$ is a word $w$ such that $y = \Network^w( x )$, $[w] = \Delta( 
x, y )$ and $w_a \ne w_b$ for all $a \ne b$, i.e. $w$ visits every vertex $v$ where $x$ and $y$ differ exactly once, and does not visit any other vertex. If there exists a geodesic from $x$ to $y$, we denote it by $x \Geodesic_\Network y$.

The set of \Define{fixed points} of $\Network$ is $\Fix( \Network ) = \{ x \in \{0,1\}^V : \Network( x ) = x \}$. It is clear that $x \in \Fix( \Network )$ if and only if $\Network^w( x ) = x$ for any word $w$, i.e. a ``parallel'' fixed point is also a ``sequential'' fixed point. The word $w$ is a \Define{fixing word} for $\Network$ \cite{AGRS20} (and we say that $w$ \Define{fixes} $\Network$) if for all $x$, $\Network^w( x ) \in \Fix( \Network )$ (see \cite{AGRS20} for some examples of fixing words). A Boolean network is \Define{fixable} if it has a fixing word.

\section{Constituencies and districts} \label{section:constituencies}

In this section, we introduce two kinds of sets of vertices, namely constituencies and districts, and we determine the complexity of some decision problems related to them. Even though both concepts will be useful to the sequel of this paper (an intuition behind the role of constituencies is given in the introduction of Section \ref{section:fixing_words_K}), we believe that the concept of constituency in particular is a natural property and is interesting to the wider graph theory community.

\subsection{Constituencies} \label{subsection:constituencies}

Let $G= (V,E)$ be a graph. A subset $S$ of $V$ is a \Define{constituency} of $G$ if there exists an independent set $I$ such that $S \subseteq \Neighbourhood(I)$ (note that this requires that $S \cap I = \emptyset$). The following are equivalent for a set of vertices $S \subseteq V$ (the proof is easy and hence omitted):
\begin{enumerate}
    \item 
    $S$ is a constituency of $G$, i.e. there exists an independent set $I$ of $G$ such that $S \subseteq \Neighbourhood(I)$;

    \item 
    $V \setminus S$ contains a maximal independent set of $G$;

    \item 
    there exists a maximal independent set $M$ of $G$ such that $M \cap S = \emptyset$;
\end{enumerate}

A \Define{non-constituency} is a set of vertices that is not a constituency. The \DecisionProblem{Constituency} (\DecisionProblem{Non-Constituency}, respectively) problem asks, given a graph $G$ and set $S$, whether $S$ a constituency (a non-constituency, respectively) of $G$.

\begin{decisionproblem}
  \problemtitle{Constituency}
  \probleminput{A graph $G = (V,E)$ and a set of vertices $S \subseteq V$.}
  \problemquestion{Is $S$ a constituency of $G$?}
\end{decisionproblem}

\begin{decisionproblem}
  \problemtitle{Non-Constituency}
  \probleminput{A graph $G = (V,E)$ and a set of vertices $S \subseteq V$.}
  \problemquestion{Is $S$ a non-constituency of $G$?}
\end{decisionproblem}

\begin{theorem} \label{theorem:complexity_constituency}
    \DecisionProblem{Constituency} is \NP-complete.
\end{theorem}

\begin{proof}
    Membership of \NP{} is known: the yes-certificate is an independent set $I$ such that $S \subseteq \Neighbourhood( I )$.

    The hardness proof is by reduction from \DecisionProblem{Set Cover}, which is \NP-complete \cite{Kar72}.
    In \DecisionProblem{Set Cover}, the input is a finite set of elements $X=\{x_1,\ldots,x_n\}$, a collection  $C=\{C_1,C_2,\ldots,C_m\}$ of subsets of $X$, and an integer $k$. The question is whether there exists a subset $Y \subseteq C$ of cardinality at most $k$ such that $\cup_{C_i\in Y}C_i = X$. 
    
    We first construct the graph $G$ on $n+mk$ vertices. $G$ consists of: vertices $Q_j=\{q_j^1,\ldots, q_j^k\}$, for each $j\in [m]$; vertices $v_i$ for each $i\in [n]$; edges from each vertex in $Q_j$ to $v_i$, whenever $x_i\in C_j$; edges connecting $\{q_1^l,q_2^l,\ldots, q_m^l\}$ in a clique, for each $l\in [k]$. Let the target set $S = \{ v_1,\ldots,v_n \}$. This concludes our construction; an illustrative example is shown in Figure \ref{fig:setcover}.

    We now show that if $(X,C,k)$ is a yes-instance of \DecisionProblem{Set Cover}, then $(G,S)$ is a yes-instance of \DecisionProblem{Constituency}.
    Let $Y \subseteq C$  be a set cover of $X$ of cardinality at most $k$. 
    We obtain the set $I$ as follows: 
    $$I=\{q_j^a:C_j \text{ is the $a$th element of $Y$}\}.$$
    Note that every vertex in $I$ exists in $G$ since $Y$ has cardinality at most $k$ (if $|Y|=k$ then the last subset to appear in $Y$ is its $k$th element exactly). Further, $I$ is an independent set, since by construction every vertex $q_j^a$ is adjacent to some other vertex $q_l^b$ if and only if $a=b$. 
    Lastly, every vertex $v_i\in Y$ is incident to some vertex in $I$; for any $i$, $\exists j: v_i\in C_j$. Then necessarily $\exists a: q_j^a\in I$, and by construction $(v_i, q_j^a)$ is an edge in $G$. 

    Conversely, if $(G,S)$ is a yes-instance of \DecisionProblem{Constituency} then $(X,C,k)$ is a yes-instance of \DecisionProblem{Set Cover}.
    Let $I$ be an independent set in $G$ which dominates $S$. By construction of $G$, $I$ has cardinality at most $k$. Suppose otherwise, for contradiction - then by the pigeon-hole principle there is some clique $C_j$ such that $|C_j\cap I|\geq 2$, contradicting that $I$ is an independent set.
    We obtain the set $Y$ of cardinality $|I|$ as follows:
    $$Y=\{C_j:\exists a \text{ such that }q_j^a\in I\}.$$
    We now show $Y$ is a set cover of $X$. For each $i\in[n]$, $v_i$ must be adjacent to some vertex in $I$; denote this vertex $q_j^a$ - now by construction $x_i$ is in the set $C_j$, and $C_j \in Y$. 

\end{proof}

\begin{figure}
    \centering
    \includegraphics[page=1,width=\textwidth]{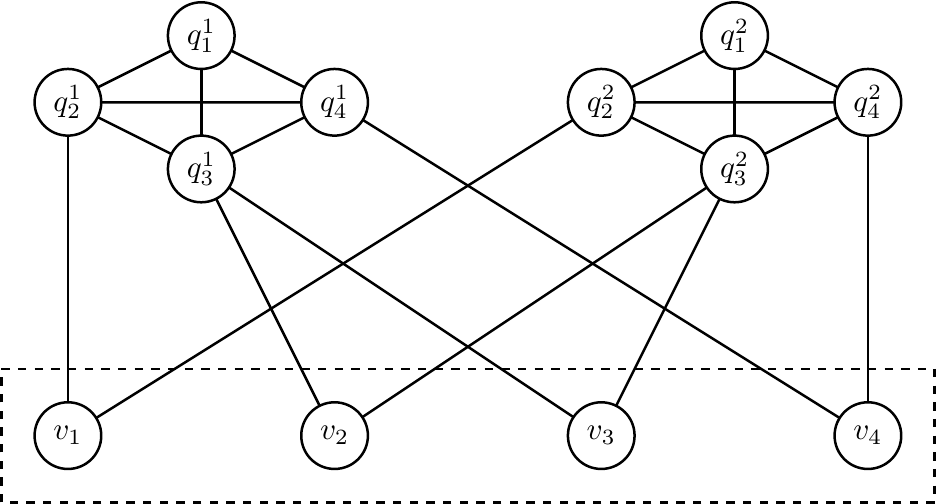}
    \caption{Illustration of the reduction from \DecisionProblem{Set Cover} to \DecisionProblem{Constituency} (the set $S$ is the vertices in the dashed box). Here the \DecisionProblem{Set Cover} instance has $C_1=\emptyset,C_2=\{x_1\},C_3=\{x_2,x_3\},C_4=\{x_4\}$, with $k=2$. Observe that both the \DecisionProblem{Set Cover} instance and the \DecisionProblem{Constituency} instance are no-instances.}
    \label{fig:setcover}
\end{figure}

\begin{corollary} \label{corollary:complexity_non-constituency}
    \DecisionProblem{Non-Constituency} is \coNP-complete.
\end{corollary}

We make four remarks about constituencies. Let $G$ be a graph, $S$ be a subset of its vertices, and $T = \Neighbourhood(S) \setminus S$. First, if $G - S$ has an isolated vertex $t$, then $S$ is a constituency of $G$ if and only if $S \setminus \Neighbourhood(t)$ is a constituency of $G - t$. Second, whether $S$ is a constituency of $G$ is independent of the edges in $G[S]$. As such, we can (and shall)  reduce ourselves to either of two canonical types of instances $(G,S)$ of \DecisionProblem{Constituency} (and of course, of \DecisionProblem{Non-Constituency} as well):

\begin{description}
    \item[Complete type:] 
    $G[S]$ is complete and $G - S$ has no isolated vertices. 

    \item[Empty type:] 
    $G[S]$ is empty and $G - S$ has no isolated vertices. 
\end{description}
Third, $S$ is a constituency of $G$ if and only if $S$ is a constituency of $G[ S \cup T]$. Therefore, we could reduce ourselves to the case where $V = S \cup \Neighbourhood(S)$; however, this assumption shall be unnecessary in our subsequent proofs and as such we shall not use it. Fourth, if $S$ is a constituency of $G$ then every subset of $S$ is also a constituency of $G$. 

\subsection{Districts} \label{subsection:districts}

A subset $T$ of vertices of a graph $G$ is a \Define{district} of $G$ if there exists $v \in V \setminus T$ such that $T \cap \Neighbourhood(v)$ is a constituency of $G - v$. A \Define{non-district} is a set of vertices that is not a district. The \DecisionProblem{District} (\DecisionProblem{Non-District}, respectively) decision problem asks, given a graph $G$ and a set $T$, whether $T$ is a district (a non-district, respectively) of $G$.

\begin{decisionproblem}
  \problemtitle{District}
  \probleminput{A graph $G = (V,E)$ and a set of vertices $T \subseteq V$.}
  \problemquestion{Is $T$ a district of $G$?}
\end{decisionproblem}

\begin{decisionproblem}
  \problemtitle{Non-District}
  \probleminput{A graph $G = (V,E)$ and a set of vertices $T \subseteq V$.}
  \problemquestion{Is $T$ a non-district of $G$?}
\end{decisionproblem}

\begin{theorem} \label{theorem:complexity_district}
\DecisionProblem{District} is \NP-complete.
\end{theorem}

\begin{proof}
Membership of \NP{} is known: the yes-certificate is a vertex $v$ and a set of vertices $I$ such that $v \notin I \cup T$, $I$ is an independent set, and $T \cap \Neighbourhood(v) \subseteq \Neighbourhood( I )$.

The hardness proof is by reduction from \DecisionProblem{Constituency}, which is \NP-complete, as proved in Theorem \ref{theorem:complexity_constituency}. Let $(G,S)$ be an instance of \DecisionProblem{Constituency}, and construct the instance $(\hat{G}, \hat{S})$ of \DecisionProblem{District} as follows.

Let $G = (V,E)$ and denote $T = V \setminus S$. Then consider a copy $T' = \{t' : t \in T\}$ of $T$ and an additional vertex $\hat{v} \notin V \cup T'$. Let $\hat{G} = (\hat{V}, \hat{E})$ with $\hat{V} = V \cup T' \cup \{ \hat{v} \}$ and $\hat{E} = E \cup \{ tt' : t \in T \} \cup \{ s\hat{v} : s \in S \}$, and $\hat{S} = S \cup T'$. This construction is illustrated in Figure \ref{fig:district}.

We only need to prove that $S$ is a constituency of $G$ if and only if $\hat{S}$ is a district of $\hat{G}$. Firstly, if $S$ is a constituency of $G$, then there exists an independent set $I$ of $G$ such that $S \subseteq \Neighbourhood( I; G )$. Then $\hat{S} \cap \Neighbourhood( \hat{v}; \hat{G} ) = S$ is contained in $\Neighbourhood( I; \hat{G} - \hat{v} )$, thus $\hat{S}$ is a district of $\hat{G}$.

Conversely, if $\hat{S}$ is a district of $\hat{G}$, then there exists $u \in \hat{V} \setminus \hat{S}$ such that $\hat{S} \cap \Neighbourhood( u; \hat{G} )$ is a constituency of $\hat{G} - u$. Then either $u = \hat{v}$ or $u \in T$. Suppose $u = t \in T$, then $t' \in \hat{S}$ is an isolated vertex of $G - t$, hence $\hat{S} \cap \Neighbourhood( t; \hat{G} )$ is not a constituency of $\hat{G} - t$. Therefore, $u = \hat{v}$ and there exists an independent set $\hat{I}$ of $\hat{G} - \hat{v}$ such that $\hat{S} \cap \Neighbourhood( \hat{v} ; \hat{G} ) = S$ is contained in $\Neighbourhood( \hat{I}; \hat{G} )$. Since $S \subseteq V$ and $\Neighbourhood(S; \hat{G} - \hat{v}) \subseteq V$, we obtain $S \subseteq \Neighbourhood( \hat{I} \cap V; \hat{G} - \hat{v} ) \cap V = \Neighbourhood( \hat{I} \cap V; G )$, where $I = \hat{I} \cap V$ is an independent set of $G$. Thus, $S$ is a constituency of $G$.

\end{proof}

\begin{figure}
    \centering
    \includegraphics[page=2,width=.5\textwidth]{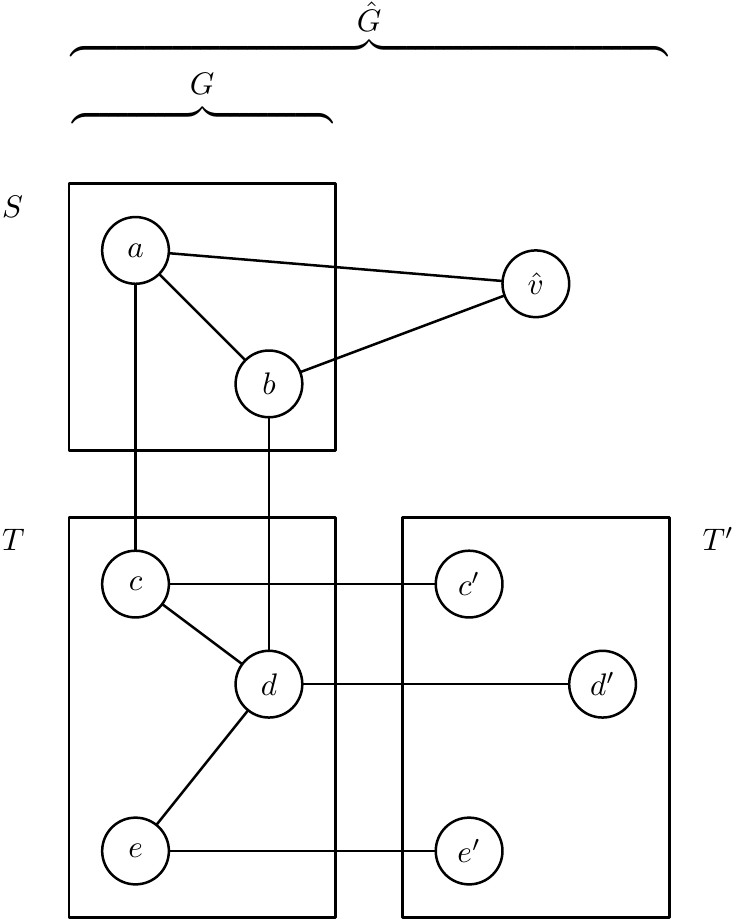}
    \caption{Example reduction from a no-instance of \DecisionProblem{Constituency} $(G,S)$ to the corresponding no-instance of \DecisionProblem{District} $(\hat{G}, \hat{S})$, with $\hat{S}:=S\cup T'$. 
}
    \label{fig:district}
\end{figure}


\begin{corollary} \label{corollary:complexity_non-district}
    \DecisionProblem{Non-District} is \coNP-complete.
\end{corollary}

If $T$ is a district of $G$, then any subset of $T$ is also a district of $G$. Therefore, any superset of a non-district is also a non-district. Furthermore, every graph $G$ has a trivial non-district, namely $V$. The \DecisionProblem{Non-Trivial Non-District} problem asks whether $G$ has any other non-district. We provide some illustrative instances in Figure \ref{fig:nondistrict_examples}. We need only consider sets $W$ with $|W|=n-1$. For $C_4$ and $C_3$, we can by symmetry assume $W=V \setminus \{a\}$, and then for $C_4$ $\{c\}$ is an independent set which dominates $\Neighbourhood(a)\cap W$, whereas for the $C_3$ there are no vertices outside $\Neighbourhood(a) \cap W$ and hence $\{b,c\}$ is a non-district. Similarly, for $P_3$, $W=V\setminus\{b\}$ is a non-trivial non-district.

\begin{decisionproblem}
    \problemtitle{Non-Trivial Non-District}
    \probleminput{A graph $G = (V,E)$.}
    \problemquestion{Does there exist a non-district $S \ne V$ of $G$?}
\end{decisionproblem}

\begin{figure}
    \centering
    \includegraphics[page=16,width=.8\textwidth]{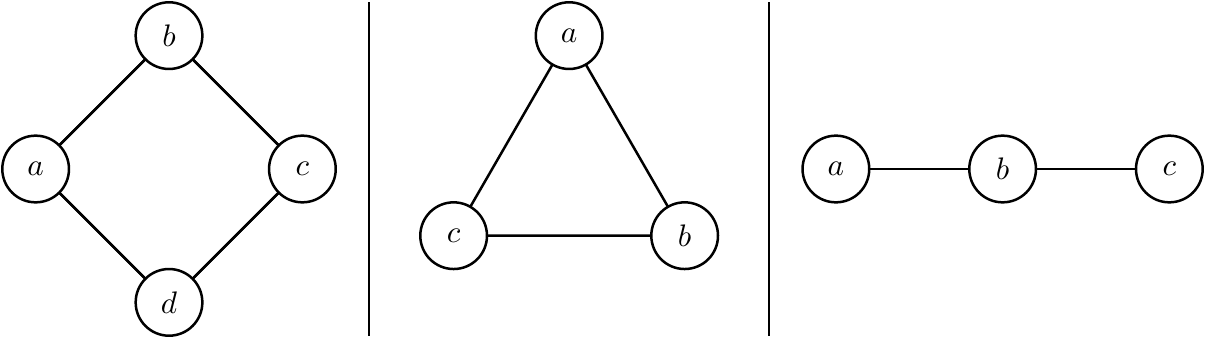}
    \caption{Some example instances of the \DecisionProblem{Non-Trivial Non-District} problem. $C_4$ (left) is a no-instance, whereas $C_3$ (center) and $P_3$ (right) are yes-instances.}
    \label{fig:nondistrict_examples}
\end{figure}

\begin{theorem} \label{theorem:complexity_non-trivial_non-district}
\DecisionProblem{Non-Trivial Non-District} is \coNP-complete.
\end{theorem}

\begin{proof}
Since any superset of a non-district is also a non-district, $G$ has a non-district $S \ne V$ if and only if there exists $v \in V$ such that $V \setminus \{ v \}$ is a non-district of $G$. 
Therefore, \DecisionProblem{Non-Trivial Non-District} is in \coNP, where the no-certificate is a collection $(I_v : v \in V)$ such that $I_v$ is an independent set of $G - v$ and $\Neighbourhood( v ) \subseteq \Neighbourhood( I_v )$ for all $v$.

The hardness proof is by reduction from \DecisionProblem{Non-Constituency}, which is \coNP-complete by Corollary \ref{corollary:complexity_non-constituency}. Let $(G = (V,E), S)$ be an instance of \DecisionProblem{Non-Constituency} of complete type (i.e. where $S$ is a clique in $G$) and denote $T = V \setminus S$. Let $V' = \{ v' : v \in V \}$ be a copy of $V$, $T'' = \{ t'' : t \in T \}$ be a second copy of $T$, and $\sigma''$ and $\hat{v}$ be two additional vertices. For any $A \subseteq V$, we denote $A' = \{ a' : a \in A \}$. Let $\hat{G} = (\hat{V}, \hat{E})$ with $\hat{V} = V \cup V' \cup T'' \cup \{ \sigma'', \hat{v} \}$ and $\hat{E} = E \cup \{ vv' : v \in V \} \cup \{ \hat{v}s, \hat{v}s', s'\sigma'' : s \in S \} \cup \{ t'' \bar{t}'', t''\sigma'' : t, \bar{t} \in T \}$. This is illustrated in Figure \ref{figure:complexity_non-trivial_non-district}.

We first show that $W_a = V \setminus \{ a \}$ is a district of $\hat{G}$ for all $a \ne \hat{v}$ (note that $W_a \cap \Neighbourhood( a; \hat{G} ) = \Neighbourhood( a; \hat{G} )$). Necessarily one of the following holds.
\begin{itemize}
    \item 
    $a = s \in S$. \\
    Then $\Neighbourhood( s; \hat{G} ) = \{ \hat{v}, s' \} \cup \Neighbourhood( s; G )$ is dominated by the independent set $\{ \sigma'' \} \cup \Neighbourhood( s; G )'$.

    \item
    $a = s' \in S'$. \\
    Then $\Neighbourhood( s'; \hat{G} ) = \{ \hat{v}, s, \sigma'' \}$ is dominated by the independent set $\{ \bar{s}, \bar{s}'\}$ where $\bar{s} \in S \setminus \{s\}$ (and so necessarily $s\bar{s} \in E$).

    \item 
    $a = \sigma''$. \\
    Then $\Neighbourhood( \sigma''; \hat{G} ) = S' \cup T''$ is dominated by the independent set $\{ \hat{v} \} \cup T'$.

    \item 
    $a = t \in T$. \\
    Then $\Neighbourhood( t; \hat{G} ) = t' \cup \Neighbourhood( t; G )$ is dominated by the independent set $\{ t'' \} \cup \Neighbourhood( t; G )'$ (or alternatively $\{t'', \hat{v}\}$). 

    \item
    $a = t' \in T'$. \\
    Then $\Neighbourhood( t'; \hat{G} ) = \{ t, t'' \}$ is dominated by the independent set $\{ \bar{t}, \bar{t}''\}$ where $t\bar{t} \in E$. (Recall $G-S$ has no isolated vertices in a \DecisionProblem{Constituency} instance of complete type.)
 
    \item 
    $a = t'' \in T''$. \\
    Then $\Neighbourhood( t''; \hat{G} ) = \{ t', \sigma'' \} \cup ( T'' \setminus \{ t'' \} )$ is dominated by the independent set $\{ t \} \cup ( V' \setminus \{ t' \} ) \cup S'$.
\end{itemize}

We now show that $W_{ \hat{v} }$ is a district of $\hat{G}$ if and only if $S$ is a constituency of $G$. We remark that $W_{ \hat{v} } \cap \Neighbourhood( \hat{v}; \hat{G} ) = S \cup S'$.
If $W_{ \hat{v} }$ is a district of $\hat{G}$, then $S \subseteq \Neighbourhood( I \setminus \Neighbourhood[ \hat{v} ; \hat{G} ]; \hat{G} )$ for some independent set $I$. Therefore $S \subseteq \Neighbourhood( I \cap T; G )$, i.e. $S$ is a constituency of $G$. Conversely, if $S$ is a constituency of $G$, say $S \subseteq \Neighbourhood( I ; G )$ for some independent set $I$ of $G$, then $I \cup \{\sigma''\}$ is an independent set of $\hat{ G }$ such that $S \cup S' \subseteq \Neighbourhood( I \cup \{\sigma''\} ; \hat{G} )$, i.e. $W_{ \hat{v} }$ is a district of $\hat{G}$.


\end{proof}

\begin{figure}
    \centering
    \includegraphics[page=5,width=\textwidth]{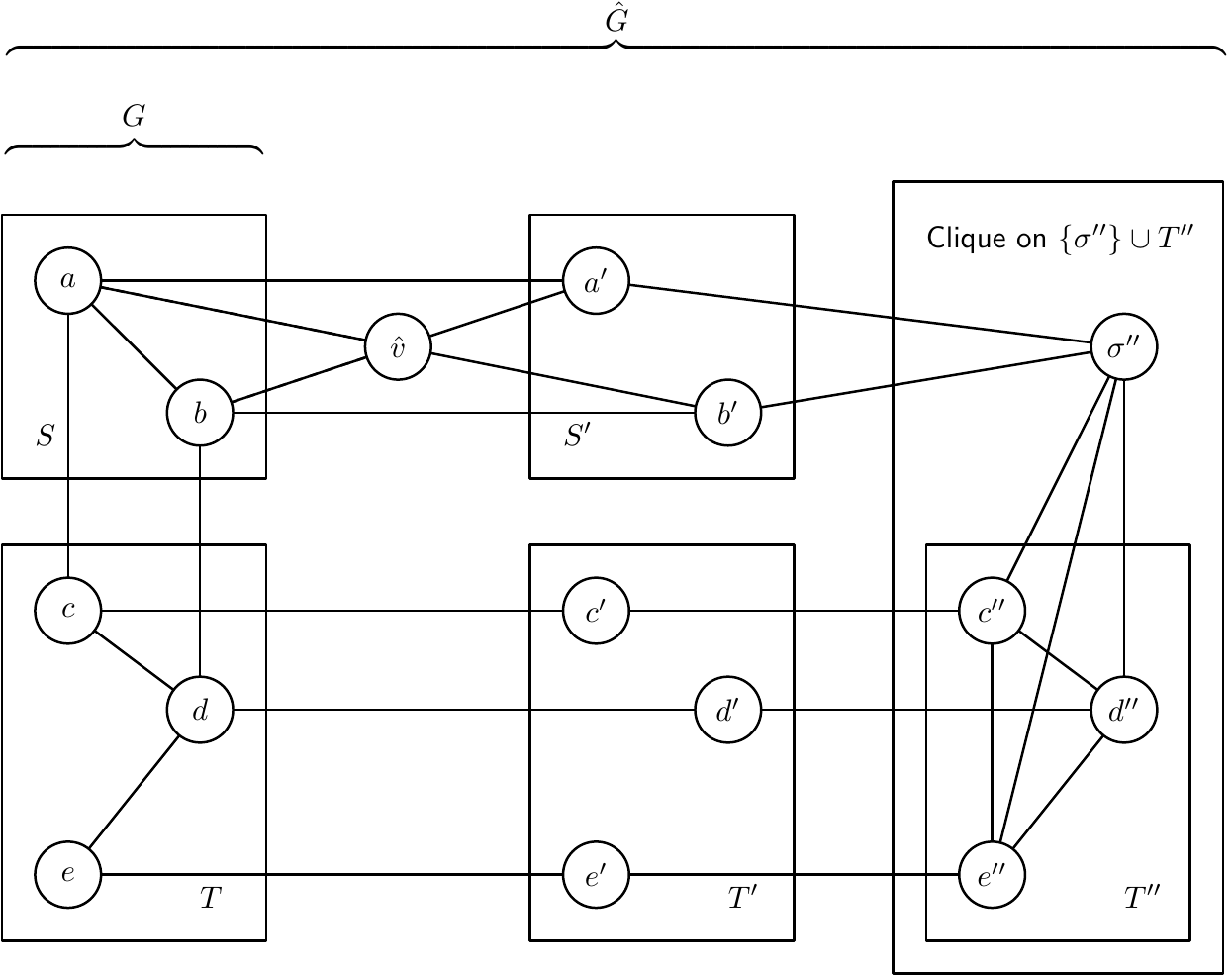}
    \caption{Illustration of the reduction from \DecisionProblem{Non-Constituency} to \DecisionProblem{Non-Trivial Non-District}.}
    \label{figure:complexity_non-trivial_non-district}.
\end{figure}

\section{Reachability of the MIS network} \label{section:reachability_K}

\subsection{The MIS network} \label{subsection:MIS_network}

By identifying a configuration $x \in \{0,1\}^V$ with its support $\one( x )$, one can interpret the MIS algorithm as sequential updates of a particular Boolean network. The \Define{MIS network} on $G$, denoted as $\MIS(G)$ or simply $\MIS$ when the graph is clear from the context, is defined by
\begin{align*}
    \MIS(x)_v &= \begin{cases}
    0 &\text{if } \exists u \in \Neighbourhood(v): x_u = 1 \\
    1 &\text{if } \forall u \in \Neighbourhood(v): x_u = 0
    \end{cases} \\
    &= \bigmeet_{u \sim v} \neg x_u,
\end{align*}
with $\MIS(x)_v = 1$ if $\Neighbourhood(v) = \emptyset$. We then have $\Fix( \MIS(G) ) = \MaximalIndependentSets(G)$ \cite{RR13,ARS14}.

The MIS algorithm then begins with the all-zero configuration, updates the state of every vertex in order, and leads to a configuration whose support is a maximal independent set. In the language of Boolean networks:
\begin{itemize}
    \item $x = 0$;

    \item $w$ is a permutation of $V$;

    \item $y = \MIS^w(x) \in \MaximalIndependentSets( G )$.
\end{itemize}

The pivotal role of constituencies for the MIS network can be explained by the equivalence below (its proof is easy and hence omitted). For a set of vertices $S \subseteq V$,
$S$ is a constituency of $G$ if and only if there exists a fixed point $y \in \MaximalIndependentSets(G)$ such that $y_S = 0$.

\subsection{Universal configurations} \label{subsection:universal_configurations}

In this paper, we are interested in removing the constraint on the initial configuration $x$. This in turn will lead to constraints on the word $w$, as we shall see in the sequel. For now, in this section, we are interested in initial configurations $x$ that can lead to any MIS $y$.

Say a configuration $x$ is \Define{$\Network$-universal} if every fixed point of $\Network$ is reachable from $x$, i.e. $x \mapsto_\Network z$ for all $z \in \Fix( \Network )$. Clearly, the all-zero configuration is $\MIS( G )$-universal, as one can reach any MIS from the empty set. In fact, those fixed points can be reached by a geodesic. We now classify the universal configurations for the MIS network, which actually also allow to reach all fixed points by a geodesic. Since the classification is based on constituencies, the problem of deciding whether a configuration is universal is \coNP-complete.

\begin{decisionproblem}
    \problemtitle{$\MIS$-Universal Configuration}
    \probleminput{A graph $G$ and a configuration $x$.}
    \problemquestion{Is $x$ an $\MIS( G )$-universal configuration?}
\end{decisionproblem}

\begin{theorem} \label{theorem:complexity_universal_configuration_K}
\DecisionProblem{$\MIS$-Universal Configuration} is \coNP-complete.    
\end{theorem}

We first characterise the configurations $y$ that are reachable from a given configuration $x$. For any configuration $x$ on $G$, we denote the collection of connected components of $G[ \one(x) ]$ as $\mathcal{C}(x)$. Before giving the full statement of the result, we provide some intuition. Suppose $y$ is reachable from $x$; we show that $y$ must satisfy two conditions. First, $y$ cannot ``create an edge'': if $[w]$ intersects an edge of $G[ \one(x ) ]$, then it will destroy it. Therefore, any edge in $G[ \one( y ) ]$ must be an (untouched) edge of $G[ \one( x ) ]$. Second, $y$ cannot ``empty out'' a connected component: in order to update a vertex $v$ from $x_v = y^{a-1}_v = 1$ to $y_v = y^a_v = 0$, there must be a neighbour $a$ of $v$ such that $y^{a-1}_v = 1$. Therefore, for any $C \in \mathcal{C}$, $y_C \ne 0$.

Proposition \ref{proposition:reachability_K} then shows that these two conditions are indeed sufficient for reachability, and in fact for reachability by a geodesic.

\begin{proposition}[Reachability for the MIS network] \label{proposition:reachability_K}
Let $G$ be a graph and $x, y$ be two configurations on $G$. The following are equivalent:
\begin{enumerate}
    \item \label{item:reachability_K_reachable}
    $x \mapsto_\MIS y$;

    \item \label{item:reachability_K_geodesic}
    $x \Geodesic_\MIS y$;

    \item \label{item:reachability_K_property}
    every edge in $G[ \one( y ) ]$ is an edge in $G[ \one( x ) ]$ and $y_C \ne 0$ for any $C \in \mathcal{C}(x)$. 
\end{enumerate}
\end{proposition}

\begin{lemma} \label{lemma:edge}
Let $x$ be a configuration on $G$, $w$ be a word, and $y = \MIS^w(x)$. If $uv$ is an edge in $G[ \one( y ) ]$, then $[w] \cap \{u, v\} = \emptyset$.
\end{lemma}

\begin{proof}
Suppose $v$ is the last updated in $\{u,v\}$, say $v = w_{a+1}$ while $w_b \notin \{ u, v \}$ for all $b > a + 1$. Then $y^a_u = 1$ and hence $y_v = y^{a+1}_v = \MIS( y^a )_v = 0$, which is the desired contradiction.
\end{proof}

\begin{proof}[Proof of Proposition \ref{proposition:reachability_K}]
Suppose that $x \mapsto_\MIS y$. It follows from Lemma \ref{lemma:edge} that every edge in $G[ \one( y ) ]$ is an edge in $G[ \one( x ) ]$. We prove that $y_C \ne 0$ for any $C \in \mathcal{C}(x)$. Suppose $y_C = 0$ for some $C \in \mathcal{C}(x)$ with $w = w_{1 : l}$ but $y^{l-1}_C \ne 0$. Then $w_l \in C$, $y^{l-1}_{C \setminus \{ w_l \}} = 0$, and $y^{l-1}_{w_l} = 1$. Since $\MIS( y^{l-1} )_{w_l} = 0$, there exists $u$ such that $u \sim w_l$ and $y^{l-1}_u = 1$.

\begin{claim}
$u \in [w_{1:l-1}]$.
\end{claim}

\begin{proof}
Firstly, since $y^{l-1}_{C \setminus \{ w_l \}} = 0$ and $y^{l-1}_u = 1$, we have $u \notin C$. Secondly, since $u \in \Neighbourhood( w_l ) \setminus C$ while $\Neighbourhood( w_l ) \cap \one( x ) \subseteq C$, we have $u \in \zero( x )$. Thirdly, since $x_u = 0$ and $y_u^{l-1} = 1$, we must have $u \in [w_{1 : l-1}]$.
\end{proof}


Finally, $uw_l$ is an edge in $G[ \one( y^{l-1} ) ]$ with $[ w_{1:l-1} ] \cap \{ u, w_l \} \ne \emptyset$, which contradicts Lemma \ref{lemma:edge}.

\medskip

Conversely, suppose that every edge in $G[ \one( y ) ]$ is an edge in $G[ \one( x ) ]$ and  $y_C \ne 0$ for any $C \in \mathcal{C}(x)$. We first describe a word $w$ and we then prove that $w$ is a geodesic from $x$ to $y$. The word $w$ is constructed in four steps as follows.
\begin{enumerate}
    \item Let $w^0$ be any permutation of $\one( y ) \cap \zero( x )$.

    \item For any $C \in \mathcal{C}(x)$, the word $w^C$ is constructed as follows. By Corollary \ref{corollary:rooted_forest}, let $T$ be a spanning forest of $C$ rooted at $D = \one( y ) \cap C$, then $w^C$ is a traversal of the spanning forest from leaves towards roots, skipping the roots. More formally, $w^C = t_1 \dots t_k$ where $\{ t_1, \dots, t_k \} = C \setminus D$ and if $t_i$ is a parent of $t_j$ in $T$, then $i > j$. 

    \item Let $w^1$ be a concatenation (in any order) of $w^C$ for every $C \in \mathcal{C}(x)$. More formally, let $\mathcal{C}(x) = \{ C_1, \dots, C_m \}$ then $w^1 = w^{C_1} \dots w^{C_m}$.

    \item Let $w = w^1 w^0$.
\end{enumerate}

We now verify that $w$ is a geodesic from $x$ to $y$, i.e. that $[w] = \Delta(x, y)$ and $y = \MIS^w(x)$. Firstly, 
\begin{align*}
    [w^0] &= \one( y ) \cap \zero( x ), \\
    [w^1] &= \bigcup_{C \in \mathcal{C}(x)} [w^C] = \bigcup_{C \in \mathcal{C}(x)} C \cap \zero( y ) = \zero( y ) \cap \one( x ), \\
    [w] &= [w^0] \cup [w^1] = \Delta( x, y ).
\end{align*}
Secondly, we prove that $\MIS^w( x )_{[w^1]} = 0$ while $\MIS^w( x )_{[w^0]} = 1$. Let $C \in \mathcal{C}$ and $w^C = t_1 \dots t_k$. By induction on $1 \le j \le k$, we have $\MIS^w( x )_{t_j} = \MIS^{w^C}( x )_{t_j} = 0$ since $t_j$ has a parent $t_i \in C$ which will only be updated after $t_j$. This shows that $\MIS^w( x )_{[w^1]} = 0$. Moreover, let $w^0 = v_1 \dots v_l$. Suppose $\MIS^w( x )_{v_i} = 0$, then let $z = \MIS^{w^1}( x )$ then there exists $u$ such that $u \sim v_i$ and $z_u = 1$. We derive a contradiction from a case analysis on $u$.
\begin{enumerate}
    \item Case 1: $y_u = 1$. \\
    Then $uv_i$ is an edge in $G[ \one(y) ]$, hence it is an edge in $G[ \one(x) ]$ so that $x_{v_i} = 1$, which is a contradiction.

    \item Case 2: $y_u = 0$ and $x_u = 0$. \\ 
    Then $u \notin \Delta( x,y ) = [w]$ hence $z_u = x_u = 0$, which is a contradiction.

    \item Case 3: $y_u = 0$ and $x_u = 1$. \\
    Then $u \in [w^1]$ and hence $z_u = y_u = 0$, which is the desired contradiction.
\end{enumerate}
Therefore $\MIS^w( x )_{[w^0]} = 1$.
\end{proof}

\begin{corollary} \label{corollary:universal_configuration_K}
The configuration $x$ is $\MIS( G )$-universal if and only if every $C \in \mathcal{C}(x)$ is a non-constituency of $G$.
\end{corollary}

\begin{proof}
If $C \in \mathcal{C}(x)$ is a constituency of $G$, then there exists a fixed point $z \in \MaximalIndependentSets(G)$ with $z_C = 0$, which is not reachable from $x$ by Proposition \ref{proposition:reachability_K}. Conversely, if every $C \in \mathcal{C}(x)$ is a non-constituency of $G$, then for any $z \in \MaximalIndependentSets( G )$ we have $z_C \ne 0$ for all $C$, and hence $z$ is reachable from $x$.
\end{proof}

In particular, the all-zero and all-one configurations are $\MIS$-universal for all graphs.

\begin{proof}[Proof of Theorem \ref{theorem:complexity_universal_configuration_K}]
Membership of \coNP{} is known: the no-certificate is a fixed point $z \in \MaximalIndependentSets(G)$ that is not reachable from $x$; checking that certificate is by finding $C \in \mathcal{C}(x)$ such that $z_C = 0$.

We prove \coNP-hardness by reduction from \DecisionProblem{Non-Constituency}. If $(G,S)$ is an instance of \DecisionProblem{Non-Constituency} of complete type, then let $x = \CharacteristicVector( S )$ so that $\mathcal{C}( x ) = \{ S \}$. By Corollary \ref{corollary:universal_configuration_K}, $x$ is universal if and only if $S$ is a non-constituency of $G$.
\end{proof}

Another consequence of Proposition \ref{proposition:reachability_K} is that any initial configuration can reach a fixed point via a geodesic.

\begin{corollary} \label{corollary:all_reach_fixed_point}
For any configuration $x$, there exists $y \in \MaximalIndependentSets( G )$ such that $x \Geodesic_\MIS y$.
\end{corollary}

\begin{proof}
Choose a vertex $v_C$ for every $C \in \mathcal{C}(x)$, then $I = \{ v_C : C \in \mathcal{C}(x) \}$ is an independent set. Let $M$ be a maximal independent set that contains $I$, then $y = \CharacteristicVector( M ) \in \MaximalIndependentSets(G)$ satisfies Property \ref{item:reachability_K_property} of Proposition \ref{proposition:reachability_K} and hence is reachable from $x$ by a geodesic.
\end{proof}

\section{Words fixing the MIS network} \label{section:fixing_words_K}

We now focus on words fixing the MIS network. As we shall prove later, every graph $G$ has a fixing word. Whether a word $w$ fixes the MIS network does not only depend on the set $[w]$ of vertices it visits. Indeed, as seen in Example \ref{example:P3} for the graph $P_3$, the word $w = abc$ does not fix $\MIS$, while $w = acb$ does fix $\MIS$. We define \DecisionProblem{Fixing Word} to be the decision problem asking, for an instance $(G, w)$, whether $w$ fixes $\MIS( G )$.

\begin{decisionproblem}
  \problemtitle{Fixing Word}
  \probleminput{A graph $G = (V,E)$ and a word $w$.}
  \problemquestion{Does $w$ fix $\MIS(G)$?}
\end{decisionproblem}

\begin{theorem} \label{theorem:complexity_fixing_word_K}
\DecisionProblem{Fixing Word} is \coNP-complete.
\end{theorem}

\DecisionProblem{Fixing Word} is in \coNP; the certificate being a configuration $x$ such that $\MIS^w( x ) \notin \MaximalIndependentSets( G )$. We shall prove that \DecisionProblem{Fixing Word} is \coNP-complete, even when restricted to permutations, in Section \ref{subsection:permises}. As such, we omit the proof of Theorem \ref{theorem:complexity_fixing_word_K}.


\subsection{Prefixing and suffixing words} \label{subsection:prefixing_suffixing}

The seminal observation is that if $G$ is a graph, and $w$ is a permutation of $V$, then $ww$ fixes $\MIS( G )$: for any initial configuration $x$, $\MIS^w( x ) \in \IndependentSets( G )$; then for any $y \in \IndependentSets( G )$, $\MIS^w( y ) \in \MaximalIndependentSets( G )$. We shall not prove this claim now, as we will prove stronger results in the sequel (see Propositions \ref{proposition:prefix} and \ref{proposition:suffix}).

Following the seminal observation above, we say that $w^\mathrm{p}$ \Define{prefixes} $\MIS( G )$ if $\MIS^{ w^\mathrm{p} }(x) \in \IndependentSets( G )$ for all $x \in \{0,1\}^V$, and that $w^\mathrm{s}$ \Define{suffixes} $\MIS( G )$ if $\MIS^{ w^\mathrm{s} }(y) \in \MaximalIndependentSets( G )$ for all $y \in \IndependentSets( G )$. In that case, for any word $\omega$, $w^\mathrm{p}\omega$ also prefixes $\MIS( G )$ and $\omega w^\mathrm{s}$ also suffixes $\MIS( G )$. Clearly, if $w = w^\mathrm{p} w^\mathrm{s}$, where $w^\mathrm{p}$ prefixes $\MIS(G)$ and $w^\mathrm{s}$ suffixes $\MIS(G)$, then $w$ fixes $\MIS(G)$. We can be more general, as shown below.

\begin{proposition} \label{proposition_prefix_and_suffix}
If $w = w_{1:l}$ where $w_{1:a}$ prefixes $\MIS(G)$, $w_{b : l}$ suffixes $\MIS( G )$, and $[w_{b : a}]$ is an independent set of $G$ for some $0 \le a, b \le l$, then $w$ fixes $\MIS( G )$.
\end{proposition}

\begin{proof}
First, suppose $a < b$, so that $w = w_1 \dots w_a \dots w_b \dots w_l$. As mentioned above, $w^\mathrm{p} = w_{1 : b-1}$ prefixes $\MIS(G)$ and $w^\mathrm{s} = w_{b : w_l}$ suffixes $\MIS(G)$, hence $w = w^\mathrm{p} w^\mathrm{s}$ fixes $\MIS(G)$.

Second, suppose $a \ge b$, so that $w = w_1 \dots w_b \dots w_a \dots w_l$. It is easily seen that for any two non-adjacent vertices $u$ and $v$, $\MIS^{vv} = \MIS^{ v }$ and $\MIS^{uv} = \MIS^{vu}$. As such, 
\[
    \MIS^w = \MIS^{w_1 \dots w_b \dots w_a \dots w_l} =
    \MIS^{w_1 \dots w_b w_b \dots w_a w_a \dots w_l} = \MIS^{w_1 \dots w_b \dots w_a w_b \dots w_a \dots w_l},
\]
and again if we let $w^\mathrm{p} = w_{1 : a}$ and $w^\mathrm{s} = w_{b : l}$, we have $\MIS^w = \MIS^{ w^\mathrm{p} w^\mathrm{s} }$, hence $w$ fixes $\MIS( G )$.
\end{proof}

We now characterise the words that prefix (or suffix) the MIS network. Interestingly, those properties depend only on $[w]$. Also, while deciding whether a word prefixes the MIS network is computationally tractable, deciding whether a word suffixes the MIS network is computationally hard as it is based on the \DecisionProblem{Non-District} problem.

\begin{proposition} 
\label{proposition:prefix}
Let $G$ be a graph. Then the word $w$ prefixes $\MIS(G)$ if and only if $[w]$ is a vertex cover of $G$.
\end{proposition}

\begin{proof}
Suppose $[w]$ is a vertex cover of $G$ and that $y = \MIS^w( x ) \notin \IndependentSets( G )$, i.e. $y_{uv} = 11$ for some edge $uv$ of $G$. Without loss, let the last update in $\{u,v\}$ be $v$, i.e. there exists $a$ such that $w_{a+1} = v$ and $w_b \notin \{ u, v \}$ for all $b > a + 1$. We obtain $y^a_u = y_u = 1$ hence $y_v = y^{a+1}_v = 0$, which is the desired contradiction.

Conversely, if $[w]$ is not a vertex cover, then there is an edge $uv \in E$ such that $[w] \cap \{u,v\} = \emptyset$. Therefore, if we take $x=\CharacteristicVector(\{u,v\})$ then $x_{uv} = 11$ and we have $y_{uv} = 11$ as well.
\end{proof}

\begin{decisionproblem}
  \problemtitle{Prefixing Word}
  \probleminput{A graph $G = (V,E)$ and a word $w$.}
  \problemquestion{Does $w$ prefix $\MIS(G)$?}
\end{decisionproblem}

\begin{corollary} \label{corollary:complexity_prefix}
\DecisionProblem{Prefixing Word} is in \P.
\end{corollary}

\begin{proposition}
\label{proposition:suffix}
Let $G$ be a graph. Then the word $w$ suffixes $\MIS(G)$ if and only if $[w]$ is a non-district of $G$.
\end{proposition}

\begin{proof}
Suppose $[w]$ is a district of $G$, i.e. there exists an independent set $I$ and a vertex $v \notin [w]$ such that $W = [w] \cap \Neighbourhood(v)$ is dominated by $I$. Let $x = \CharacteristicVector( I )$ (in particular, $x_v=0$), and let $y = \MIS^w(x)$. Then for any $u \in W$, $u$ has a neighbour in $I$, hence $y_u = 0$; thus $y_{\Neighbourhood[v]} = 0$ and $w$ does not suffix $\MIS$.

Conversely, suppose $w$ does not suffix $\MIS(G)$, i.e. there exists $x \in \IndependentSets(G)$ and $v \in V$ such that $y = \MIS^w(x)$ with $y_{ \Neighbourhood[ v ] } = 0$. By Proposition \ref{proposition:reachability_K}, $y \in \IndependentSets( G )$ and $y \ge x$, hence $x_{ \Neighbourhood[ v ] } = 0$. Let $W = [w] \cap \Neighbourhood(v)$ and $I = \one( y ) \cap \Neighbourhood(W)$; we note that $I$ is an independent set. For each $u \in W$, we have $y_u = 0$ hence there exists $i \in V$ such that $y_i = 1$ and $u \in \Neighbourhood(i)$, and hence $i \in I$. Therefore, $W \subseteq \Neighbourhood(I)$ and $W$ is a constituency of $G - v$, i.e. $[w]$ is a district of $G$.
\end{proof}

\begin{decisionproblem}
  \problemtitle{Suffixing Word}
  \probleminput{A graph $G = (V,E)$ and a word $w$.}
  \problemquestion{Does $w$ suffix $\MIS(G)$?}
\end{decisionproblem}

\begin{corollary}
\DecisionProblem{Suffixing Word} is \coNP-complete.
\end{corollary}

\begin{proof}
This immediately follows from Theorem \ref{theorem:complexity_district}.
\end{proof}

\subsection{Fixing sets} \label{subsection:fixing_sets}

Some graphs have fixing words that do not visit all vertices. For instance, if $G = K_n$ is the complete graph with vertices $v_1, \dots, v_n$, then it is easily shown that $w = v_1 \dots v_{n-1}$ is a fixing word for the MIS network. In general, we say that a set $S$ of vertices of $G$ is a \Define{fixing set} of $G$ if there exists a word $w$ with $[w] = S$ that fixes $\MIS(G)$. 

We first characterise the fixing sets of graphs. Interestingly, those are the same sets $S$ such that $ww$ is a fixing word of $\MIS(G)$ for any permutation $w$ of $S$.

\begin{proposition} \label{proposition:fixing_sets}
Let $S$ be a subset of vertices of $G$. The following are equivalent.
\begin{enumerate}
    \item \label{item:word_fixing}
    $S$ is a fixing set of $\MIS( G )$, i.e. there exists a fixing word $w$ of $\MIS(G)$ with $[w] = S$.

    \item \label{item:words_prefix_suffix}
    For all words $w^\mathrm{p}$, $w^\mathrm{s}$ such that $[ w^\mathrm{p} ] = [ w^\mathrm{s} ] = S$, the word $w^\mathrm{p} w^\mathrm{s}$ fixes $\MIS( G )$.

    \item \label{item:subset}
    $S$ is a vertex cover and a non-district.
\end{enumerate}
\end{proposition}

\begin{proof}
$\ref{item:word_fixing} \implies \ref{item:subset}$. Since $w$ prefixes $\MIS( G )$, $S = [w]$ is a vertex cover by Proposition \ref{proposition:prefix}; similarly, since $w$ suffixes $\MIS( G )$, $S = [w]$ is a non-district by Proposition \ref{proposition:suffix}. 

$\ref{item:subset} \implies \ref{item:words_prefix_suffix}$. Since $S$ is a vertex cover, then by Proposition \ref{proposition:prefix} $w^\mathrm{p}$ prefixes $\MIS( G )$; similarly, by Proposition \ref{proposition:suffix} $w^\mathrm{s}$ suffixes $\MIS( G )$. Therefore, $w^\mathrm{p} w^\mathrm{s}$ fixes $\MIS( G )$.

$\ref{item:words_prefix_suffix} \implies \ref{item:word_fixing}$. Trivial.
\end{proof}

The \DecisionProblem{Fixing Set} problem asks, given a graph $G$ and a set of vertices $S$, if $S$ is a fixing set of $G$. In other words, it asks whether the vertices outside of $S$ can be skipped by some fixing word.

\begin{decisionproblem}
  \problemtitle{Fixing Set}
  \probleminput{A graph $G = (V,E)$ and a set $S \subseteq V$.}
  \problemquestion{Is $S$ a fixing set of $\MIS(G)$?}
\end{decisionproblem}

\begin{theorem} \label{theorem:complexity_fixing_set}
\DecisionProblem{Fixing Set} is \coNP-complete.
\end{theorem}

\begin{proof}
Membership of \coNP{} is known: the no-certificate is a permutation $w$ of $S$ and an initial configuration $x \in \{0,1\}^V$ such that $\MIS^{ww}(x) \notin \MaximalIndependentSets(G)$ (by Proposition \ref{proposition:fixing_sets}).

The hardness proof is by reduction from \DecisionProblem{Non-District}, which is \coNP-complete, as proved in Theorem \ref{theorem:complexity_district}. Let $(G, S)$ be an instance of \DecisionProblem{Non-District}, and construct the instance $( \hat{G}, \hat{S} )$ of \DecisionProblem{Fixing Set} as follows. 

Let $G = (V,E)$ and $T = V \setminus S$. For any $t \in T$, let $G_t = (V_t \cup \{ \hat{t} \}, E_t)$ be the graph defined as follows: $V_t = \{ u_t : u \in V \setminus t \}$  is a copy of all the vertices apart from $t$, which is replaced by a new vertex $\hat{t} \notin V_t$, and $E_t = \{ a_tb_t : ab \in E, a,b \ne t \} \cup \{ s_t \hat{t} : st \in E, s \in S \}$ is obtained by removing the edges between $\hat{t}$ and the rest of $T$. Then $G$ is the disjoint union of all those graphs, i.e. $G = \bigcup_{t \in T} G_t$, while $\hat{S} = \bigcup_{t \in T} V_t$. For the sake of simplicity, we shall use the notation $A_t = \{ u_t : u \in A \}$ for all $A \subseteq V \setminus \{ t \}$. Our construction is illustrated in Figure \ref{fig:complexity_fixing_set}.

By construction, $\hat{G} - \hat{S}$ is the empty graph on $\{ \hat{t} : t \in T \}$, hence $\hat{S}$ is a vertex cover of $\hat{G}$. All we need to show is that $\hat{S}$ is a non-district of $\hat{G}$ if and only if $S$ is a non-district of $G$. We have that $\hat{S}$ is a district of $\hat{G}$ if and only if there exists $\hat{t}$ and an independent set $\hat{I}$ of $\hat{G} - \hat{t}$ such that $W = \hat{S} \cap \Neighbourhood( \hat{t}; \hat{G} ) = ( S \cap \Neighbourhood( t; G ) )_t$ is contained in $\Neighbourhood( \hat{I} ; \hat{G} )$. We have $\hat{I} \cap V_t = I_t$ for some independent set $I$ of $G$. Since $W \subseteq V_t$ and $\Neighbourhood( W; \hat{G} - \hat{t} ) \subseteq V_t$, we have $W \subseteq \Neighbourhood( \hat{I} \cap V_t; \hat{G} ) \cap V_t = \Neighbourhood( I; G )_t$, which is equivalent to $S$ being a district of $G$.
\end{proof}

\begin{figure}
    \centering
    \includegraphics[page=3, width=\textwidth]{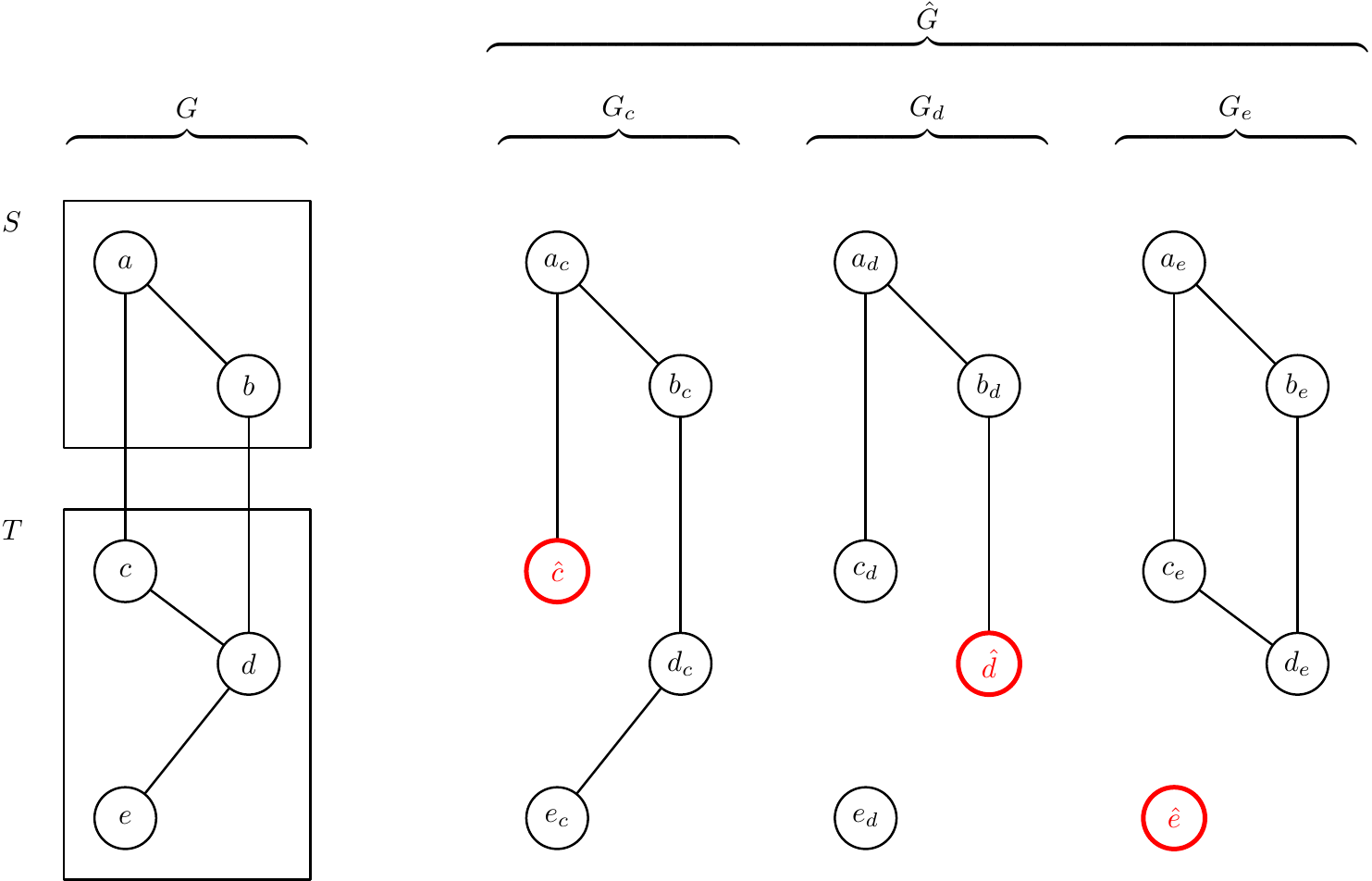}
    \caption{Example reduction from a no-instance of \DecisionProblem{Non-District} $(G,S)$ to the corresponding no-instance of \DecisionProblem{Fixing Set} $(\hat{G}, \hat{S})$, with $\hat{S} = V_c \cup V_d \cup V_e$.}
    \label{fig:complexity_fixing_set}
\end{figure}

Clearly, if $S$ is a fixing set of $\MIS(G)$, then every superset of $S$ is also a fixing set. Moreover, every graph $G$ has a trivial fixing set, namely $V$. The \DecisionProblem{Non-Trivial Fixing Set} asks whether $G$ has any other fixing set. Equivalently, it asks whether any vertex can be skipped by a fixing word.

\begin{decisionproblem}
    \problemtitle{Non-Trivial Fixing Set}
    \probleminput{A graph $G$.}
    \problemquestion{Does there exist a fixing set $S \ne V$ of $G$?}
\end{decisionproblem}

\begin{theorem} \label{theorem:complexity_non-trivial_fixing_set}
    \DecisionProblem{Non-Trivial Fixing Set} is \coNP-complete.
\end{theorem}

\begin{proof}
We prove that $G$ has a non-trivial fixing set if and only if it has a non-trivial non-district. If $G$ has a non-trivial fixing set, then there exists $S \ne V$ which is a vertex cover and a non-district of $G$, hence $S$ is a non-trivial non-district of $G$. Conversely, if $G$ has a non-trivial non-district, then there exists $v$ such that $S = V \setminus \{ v \}$ is a non-district of $G$, in which case $S$ is also a vertex cover, and hence a non-trivial fixing set of $G$.

The \coNP-completeness of \DecisionProblem{Non-Trivial Fixing Set} then follows Theorem \ref{theorem:complexity_non-trivial_non-district}. The connection between the two problems is illustrated in Figure \ref{fig:non-trivial_fixing_set}.
\end{proof}

\begin{figure}
    \centering
    \includegraphics[page=17,width=.25\textwidth]{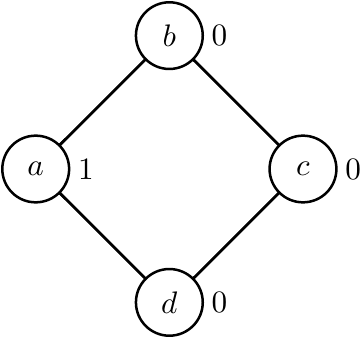}
    \caption{$C_4$ is a no-instance of the \DecisionProblem{Non-Trivial Non-District} problem, and hence also a no-instance of \DecisionProblem{Non-Trivial Fixing Set}. For any word $w$ with $[w]=\{b,c,d\}$, $\MIS^w(1000)=1000\notin \MaximalIndependentSets( C_4 )$. By symmetry, no set of three vertices is a fixing set for $\MIS(C_4)$.}
    \label{fig:non-trivial_fixing_set}
\end{figure}



\subsection{Permises} \label{subsection:permises}

The MIS algorithm doesn't use any word $w$ to update the state of each vertex, but instead restricts itself to $w$ being a permutation of $V$. 
As such, we now focus on permutations and we call a permutation of $V$ that fixes $\MIS(G)$ a \Define{permis} of $G$. The \DecisionProblem{Permis} decision problem is equivalent to the \DecisionProblem{Fixing Word} problem, restricted to permutations.

\begin{decisionproblem}
    \problemtitle{Permis}
    \probleminput{An undirected graph $G = (V,E)$ and a permutation $w$ of $V$.}
    \problemquestion{Is $w$ a permis of $G$?}
\end{decisionproblem}

Let $w$ be a permutation of $V$, then $w$ naturally induces a linear order on $V$, whereby $w_i \succ w_j$ whenever $i < j$, i.e. $w_i$ is updated before $w_j$. Then consider the orientation of $G$ induced by $w$: $G^w = (V, E^w)$ with $E^w = \{ (u, v) : uv \in E, u \succ v \}$. We see that $G^w$ is acyclic, and that conversely any acyclic orientation of $G$ is given by some $G^w$. A simple application of \cite[Theorem 1]{AGMS09} shows that if $w, w'$ are two permutations of $V$ such that $G^w = G^{w'}$, then $w$ is a permis if and only if $w'$ is a permis.


We say that the vertex $v$ is \Define{covered} by $w$ if for every $x \in \{ 0, 1 \}^V$, $y_{\Neighbourhood[ v ]} \ne 0$, where $y = \MIS^w(x)$. Thus, $w$ is a permis if and only if $w$ covers all vertices.

\begin{decisionproblem}
  \problemtitle{Covered Vertex}
  \probleminput{A graph $G = (V,E)$, a permutation $w$ of $V$ and a vertex $v \in V$.}
  \problemquestion{Is $v$ covered by $w$?}
\end{decisionproblem}

We now give a sufficient condition for a vertex to be covered. Let $G$ be a graph, $H$ be an orientation of $G$, and let $t$ and $v$ be vertices of $G$. We say $t$ is \Define{transitive} if for all $a, b \in V$, $t \to a \to b$ implies $t \to b$ in $G^w$. We say $v$ is \Define{near-transitive} if there exists a transitive vertex $t$ such that $\Neighbourhood[ t; G ] \subseteq \Neighbourhood[ v; G ]$.

\begin{lemma} \label{lemma:near-transitive_covered}
If $v$ is a near-transitive vertex of $G^w$, then $v$ is covered by $w$.
\end{lemma}

\begin{proof}
First, suppose $v = t$ is transitive. Let $x$ such that $y_{ \Neighbourhood[ t; G ] } = 0$. We shall repeatedly use the fact that for any vertex $u$, if $y_{ \InNeighbourhood[ u; G^w ] } = 0$, then there exists $u' \in \OutNeighbourhood( u; G^w ) \cap \one(x)$, i.e. $x_{u'} = 1$ and $u \to u'$ ($u$ is updated after $u$). 

Since $y_{ \InNeighbourhood[ t; G^w ] } = 0$, there exists $a \in \OutNeighbourhood( t; G^w ) \cap \one(x)$. Without loss let $a$ be the last vertex of this kind to be updated: if $a' \ne a$ satisfies $a' \in \OutNeighbourhood( t; G^w ) \cap \one( x )$, then $a' \to a$. Again, since $x_a = 1$, we have $y^{ \InNeighbourhood( a; G^w ) } = 0$; and since $y_a = 0$ as well, there exists $b \in \OutNeighbourhood( a; G^w ) \cap \one( x )$. Since $t \to a \to b$ and by transitivity of $t$, we obtain $t \to b$, but then $b \in \OutNeighbourhood( t; G^w ) \cap \one( x )$ and hence $b \to a$, which is the desired contradiction.

Second, suppose that $t$ is transitive (and hence, as shown above, covered) and that $\Neighbourhood[ t; G ] \subseteq \Neighbourhood[ v; G ]$. For all $x$ we have $y_{ \Neighbourhood[ t; G ] } \ne 0$, and hence $y_{ \Neighbourhood[ v; G ] } \ne 0$, thus $v$ is also covered by $w$.
\end{proof}

\begin{theorem} \label{theorem:complexity_permis}
\DecisionProblem{Permis} is \coNP-complete.
\end{theorem}

\begin{proof}
Membership of \coNP{} is known: the no-certificate is a configuration $x$ such that $y = \MIS^w( x ) \notin \MaximalIndependentSets(G)$.

The hardness proof is by reduction from \DecisionProblem{Non-Constituency}, which is \coNP-complete by Corollary \ref{corollary:complexity_non-constituency}. Let $(G,S)$ be an instance of \DecisionProblem{Non-Constituency} of empty type and construct the instance $( \hat{G}, w )$ of \DecisionProblem{Permis} as follows. Let $T = V \setminus S$ and $T' = \{ t' : t \in  T \}$ be a copy of $T$. Then let $\hat{G}$ be the graph with vertex set $\hat{V} = \{ v, a, b \} \cup V \cup T'$, and with edges $\hat{E} = E \cup \{ sv : s \in S \} \cup \{ va, ab \} \cup \{ tt' : t \in T \}$. Let $w$ be a permutation of $\hat{V}$ such that $v \succ a \succ b \succ T \succ T' \succ S$. This is illustrated in Figure \ref{figure:complexity_permis}.

We claim that $w$ is a permis of $\hat{G}$ if and only if $S$ is not a constituency of $G$. Firstly, the vertices in $S \cup T' \cup \{ b \}$ are all transitive and hence the vertices in $T \cup \{ a \}$ are near-transitive. Therefore, $w$ is a permis if and only if $v$ is covered. We prove that $v$ is covered if and only if $S$ is not a constituency of $G$.

If $S$ is a constituency of $G$, then let $I \subseteq T$ be a maximal independent set of $G$ (and hence an independent set of $\hat{G}$ as well) such that $S \subseteq \Neighbourhood(I)$.
Let $x = \CharacteristicVector( I \cup \{ a, b \} )$.
Then $y_v = 0$ (because $x_a = 1)$, $y_a = 0$ (because $x_b = 1$), $y_I = 1$ and $y_S = 0$ (because for any vertex $u$, if $x_u = 1$ and $x_{ \Neighbourhood(u) } = 0$, then $y_u = 1$ and $y_{ \Neighbourhood(u) } = 0$). Thus $y_{ \Neighbourhood[v] } = 0$.

Conversely, if $y_{ \Neighbourhood[v] } = 0$, then for any $s \in S$, $y_{ \Neighbourhood(s) } \ne 0$. Since $y_v = 0$, there is $t \in T$ such that $ts \in E$ and $y_t = 1$. Therefore, the set $\one( y ) \cap T$ is an independent set that dominates $S$, i.e. $S$ is a constituency of $G$.
\end{proof}

\begin{figure}
    \centering
    \includegraphics[page=6,width=0.8\textwidth]{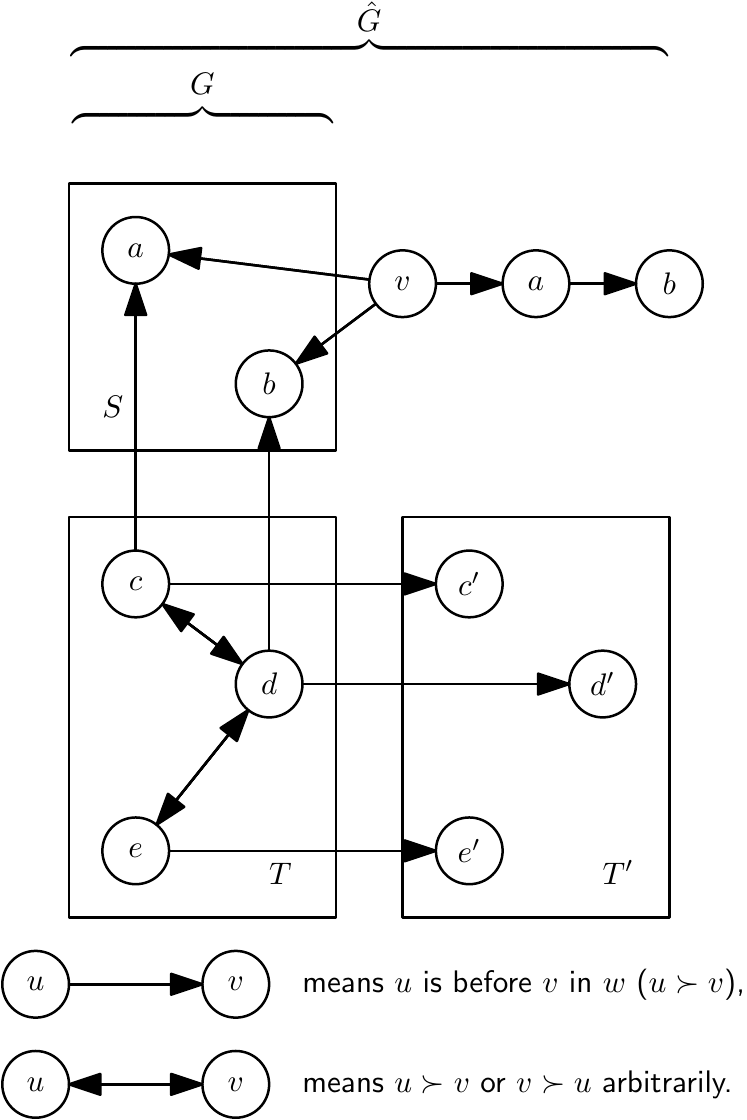}
    \caption{Illustration of the reduction from \DecisionProblem{Non-Constituency} to \DecisionProblem{Permis}.}
    \label{figure:complexity_permis}
\end{figure}

The proof of Theorem \ref{theorem:complexity_permis} also settles the complexity of \DecisionProblem{Covered Vertex}.

\begin{theorem} \label{theorem:complexity_covered_vertex}
\DecisionProblem{Covered Vertex} is \coNP-complete.
\end{theorem}

\section{Permissible and non-permissible graphs} \label{section:permissible_non-permissible}

We say that $G$ is \Define{permissible} if it has a permis. As we shall see, not all graphs are permissible. In this subsection, we exhibit permissible and non-permissible graphs, and we prove that deciding whether a graph is permissible is computationally hard.


We classified (non-)permissible graphs by computer search, using \verb|nauty|'s \verb|geng| utility (\url{https://doi.org/10.1016/j.jsc.2013.09.003}) to exhaustively generate connected graphs up to 9 vertices. Full results are available at \url{https://github.com/dave-ck/MISMax/}. Here are some highlights. 
Of 273194 connected graphs on at most nine vertices, only 432 are non-permissible; the heptagon $C_7$, 13 8-vertex graphs (including the perfect graph shown in Figure \ref{figure:permisless_perfect_graph}), and 418 9-vertex graphs. The Petersen graph is also non-permissible. 


We prove the graph in Figure \ref{figure:permisless_perfect_graph} is perfect as follows. First note that four vertices in the graph have degree $3$ and four vertices in the graph have degree $5$. The absence of an induced $C_5$ can be verified manually. There is no induced $C_7$: any subgraph on seven vertices includes at least one vertex formerly of degree $5$ hence of degree at least $4$ in the induced subgraph. Similarly, there is no induced $\overline{C_7}$; any subgraph on seven vertices includes at least one vertex formerly of degree $3$ and hence of degree at most $3$ in the induced subgraph ($\overline{C_7}$ is $4$-regular). 

\begin{figure}
    \centering
    \includegraphics[page=15,width=0.3\textwidth]{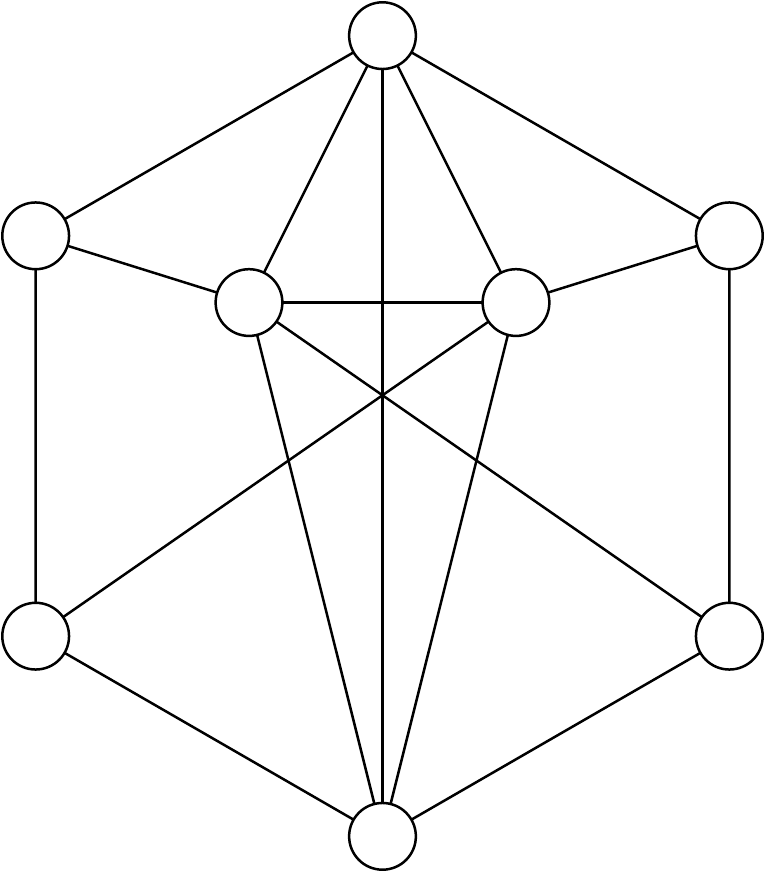}
    \caption{An 8-vertex perfect graph with no permis.}
    \label{figure:permisless_perfect_graph}
\end{figure}


\subsection{Permissible graphs} \label{subsection:permissible_graphs}

We now exhibit large classes of permissible graphs.

An orientation of $G$ is \Define{transitive} (\Define{near-transitive}, respectively) if all the vertices are transitive (near-transitive, respectively). Any transitive orientation is necessarily acyclic. A graph that admits a transitive orientation is called a \Define{comparability graph}. The following are comparability graphs: complete graphs, bipartite graphs, permutation graphs, cographs, and interval graphs. Accordingly, we say that a graph that admits a near-transitive orientation is a \Define{near-comparability graph}. Lemma \ref{lemma:near-transitive_covered} immediately yields the permissibility of near-comparability graphs.

\begin{proposition} \label{proposition:near-comparability_permissible}
All near-comparability graphs are permissible.
\end{proposition}

We now give a characterisation of near-comparability graphs below. Recall that the vertex $m$ is a benjamin of $G$ if there is no vertex $v$ with $\InNeighbourhood[ v ] \subset \InNeighbourhood[ m ]$ and that $G_\Benjamins$ is the subgraph of $G$ induced by its benjamins.

\begin{theorem} \label{theorem:near-comparability}
Let $G$ be a graph. The following are equivalent:
\begin{enumerate}
    \item \label{item:near-comparability}
    $G$ is a near-comparability graph, i.e. it admits a near-transitive orientation;

    \item \label{item:near-transitive_acyclic}
    $G$ admits a near-transitive acyclic orientation;

    \item \label{item:benjamins_comparability}
    $G_\Benjamins$ is a comparability graph.
\end{enumerate}
\end{theorem}

\begin{proof}
$\ref{item:near-comparability} \implies \ref{item:near-transitive_acyclic}$. Let $H$ be a near-transitive orientation of $G$. Let $T$ be the set of transitive vertices of $H$. Then $T$ is disjoint from all the cycles in $H$ (for if $t \in T$ belongs to the cycle $t \to v_1 \to \dots \to v_k \to t$, we must have $t \to v_k$ by transitivity), and in particular $H[T]$ is acyclic. Construct any acyclic orientation of $G$, say $G^w$, such that $G^w[ T ] = H[ T ]$ and $V \setminus T \succ T$. Note such an orientation can be constructed efficiently (for example by a greedy algorithm). Then any vertex in $T$ is still transitive in $G^w$, and hence $G^w$ is near-transitive.


$\ref{item:near-transitive_acyclic} \implies \ref{item:benjamins_comparability}$. Let $G$ be a near-comparability graph. In order to prove that $G_\Benjamins$ is a comparability graph, we need to consider the graph $\Twins{G}$ obtained by only keeping one twin out of every set of twins. More formally, for any $v \in V(G)$, let $\Twins{v} = \{ u \in V : \Neighbourhood[u] = \Neighbourhood[v] \}$ be the equivalence class of $v$. We further denote $\Twins{V} = \{ \Twins{v} : v \in V \}$, $ \Twins{E} = \{ \Twins{v} \Twins{v'} : vv' \in E \}$. Then $\Twins{G} = ( \Twins{V}, \Twins{E}  )$. 

Say that a graph is \Define{closed twin-free} if there are no closed twins, i.e. $\Neighbourhood[ u ] \ne \Neighbourhood[ v ]$ for all $u \ne v \in V$.

\begin{claim} \label{claim:twin-free}
$\Twins{G}$ is closed twin-free.
\end{claim}

\begin{proof}
Suppose $\Neighbourhood[ \Twins{u}; \Twins{G} ] = \Neighbourhood[ \Twins{v}; \Twins{G} ]$. Then $\Neighbourhood[ u; G ] = \bigcup_{ \Twins{a} \in \Neighbourhood[ \Twins{u}; \Twins{G} ] } \Twins{a} = \Neighbourhood[ v; G ]$, hence $u$ and $v$ are closed twins in $G$ and $\Twins{u} = \Twins{v}$.
\end{proof}

\begin{claim} \label{claim:twin_near-comparability}
$\Twins{G}$ is a near-comparability graph.
\end{claim}

\begin{proof}
Suppose $G^w$ is a near-transitive acyclic orientation of $G$. Then for every equivalence class $c \in \Twins{V}$, $G^w[ c ]$ is a transitive tournament with a unique source.  Consider the orientation $H$ of $\Twins{G}$ naturally induced by $w$, i.e. $\Twins{u} \to \Twins{v}$ in $\Twins{G}$ if and only if $u \to v$, where $u$ and $v$ are the unique sources of $G^w[ \Twins{u} ]$ and $G^w[ \Twins{v} ]$, respectively. We now prove that all vertices of $\Twins{G}$ are near-transitive in $H$. Firstly, if $t$ is transitive in $G^w$, then $\Twins{t}$ is transitive in $H$. Secondly, if $v$ satisfies $\Neighbourhood[ v; G ] \supseteq \Neighbourhood[ t; G ]$ for some transitive $t$, then $\Neighbourhood[ \Twins{v}; \Twins{G} ] \supseteq \Neighbourhood[ \Twins{t}; \Twins{G} ]$ and hence $\Twins{v}$ is near-transitive in $H$.
\end{proof}

\begin{claim} \label{claim:twin_benjamins_comparability}
$\Twins{G}_\Benjamins$ is a comparability graph. 
\end{claim}

\begin{proof}
Let $\Twins{G}^w$ be a near-transitive acyclic orientation of $\Twins{G}$ and let $B = \Benjamins( \Twins{G} )$. If $m \in B$, then there exists a transitive vertex $t$ such that $\Neighbourhood[ m ] \supseteq \Neighbourhood[ t ]$. Since $m \in B$ we have $\Neighbourhood[ m ] = \Neighbourhood[ t ]$ and since $\Twins{G}$ is closed twin-free we obtain $m = t$, i.e. $m$ is transitive in $\Twins{G}^w$. Therefore, $m$ is also transitive in $\Twins{G}^w[ B ] = { \Twins{G}_\Benjamins }^{w[ B ]}$. Thus, ${ \Twins{G}_\Benjamins }^{w[ B ]}$ is a transitive orientation of $\Twins{G}_\Benjamins$.
\end{proof}

\begin{claim} \label{claim:benjamins_comparability}
$G_\Benjamins$ is a comparability graph.
\end{claim}

\begin{proof}
We first remark that $\Twins{u} \in \Benjamins( \Twins{G} )$ if and only if $u \in \Benjamins(G)$. Now, let ${ \Twins{G}_\Benjamins }^{ w' }$ be a transitive orientation, then consider the orientation ${ G_\Benjamins }^w$ of $G_\Benjamins$ as follows. First, fix an arbitrary order of every equivalence class $\Twins{v}$. Second, orient $u \to v$ if $\Twins{u} \to \Twins{v}$ in ${ \Twins{G}_\Benjamins }^{ w' }$. 

We now verify that this orientation is transitive. If $u \to a \to b$ in ${ G_\Benjamins }^w$, then $\Twins{u} \to \Twins{a} \to \Twins{b}$ in ${ \Twins{G}_\Benjamins }^{ w' }$, hence $\Twins{u} \to \Twins{b}$ in ${ \Twins{G}_\Benjamins }^{ w' }$, and finally $u \to b$ in ${ G_\Benjamins }^w$.
\end{proof}

$\ref{item:benjamins_comparability} \implies \ref{item:near-comparability}$. Construct the orientation of $G$ as follows. First, use the transitive orientation on $G_\Benjamins$. Second, orient every edge $v \to m$ where $v \notin \Benjamins( G )$ and $m \in \Benjamins( G )$. Third, use any orientation on $G - \Benjamins( G )$. Then the vertices in $\Benjamins( G )$ remain transitive, and for any $v \notin \Benjamins( G )$, there exists $m \in \Benjamins( G )$ such that $\Neighbourhood[ m; G ] \subseteq \Neighbourhood[ v; G ]$, i.e. $v$ is near-transitive.

\end{proof}

Recognising comparability graphs can be done in polynomial time; see \cite{MS97} and references therein. In fact, the algorithm in \cite{MS97} not only decides whether a graph is a comparability graph, but it also produces a transitive orientation if such exists. In view of Theorem \ref{theorem:near-comparability}, applying that algorithm to $G_\Benjamins$ not only decides whether a graph is a near-comparability graph, but it also produces a near-transitive orientation (see the proof of Claim \ref{claim:benjamins_comparability}) if one exists.


Any induced subgraph of a comparability graph is a comparability graph. However, as we shall prove below, any graph is the induced subgraph of some near-comparability graph. Thus, Proposition \ref{proposition:near-comparability_permissible} shows that every graph is the induced subgraph of a permissible graph. This entails that the class of permissible graphs is not hereditary, i.e. it is impossible to characterize permissible graphs by some forbidden induced subgraphs.

\begin{corollary} \label{corollary:induced_permissible}
For every graph $G$, there exists a near-transitive (and hence permissible) graph $H$ such that $G$ is an induced subgraph of $H$.
\end{corollary}

\begin{proof}
Let $G = (V,E)$ and construct the graph $\hat{G} = (\hat{V}, \hat{E})$ so that $G = \hat{G}[ V ]$ as follows. Let $V' = \{ v' : v \in V \}$ be a copy of $V$, $\hat{V} = V \cup V'$, and $\hat{E} = E \cup \{ vv' : v \in V \}$. Then $\hat{G}$ is a near-comparability graph: let $w$ be a permutation of $\hat{V}$ where $V \succ V'$, then $\hat{G}^w$ is a near-transitive acyclic orientation of $\hat{G}$. By Proposition \ref{proposition:near-comparability_permissible}, $H$ is permissible.
\end{proof}

We now introduce an operation on graphs, that we call \Define{graph composition}, that preserves permissibility. Let $H$ be an $n$-vertex graph, $G_1, \dots, G_n$ other graphs, then the composition $H(G_1, \dots, G_n)$ is obtained by replacing each vertex $h$ of $H$ by the graph $G_h$, and whenever $hh' \in E(H)$, adding all edges between $G_h$ and $G_{h'}$. More formally, we have
\begin{align*}
    V(G) &= \{ v_h :  h \in V(H), v \in V(G_h) \}, \\
    E(G) &= \{ v_h v'_{h'} : hh' \in E(H), v \in V( G_h ), v' \in V( G_{h'} ) \} \cup \{ v_h v'_h : h \in V(H), vv' \in E( G_h ) \}.
\end{align*}
This is illustrated in Figure \ref{figure:composition}.

\begin{figure}
    \centering
    \includegraphics[page=11,width=\textwidth]{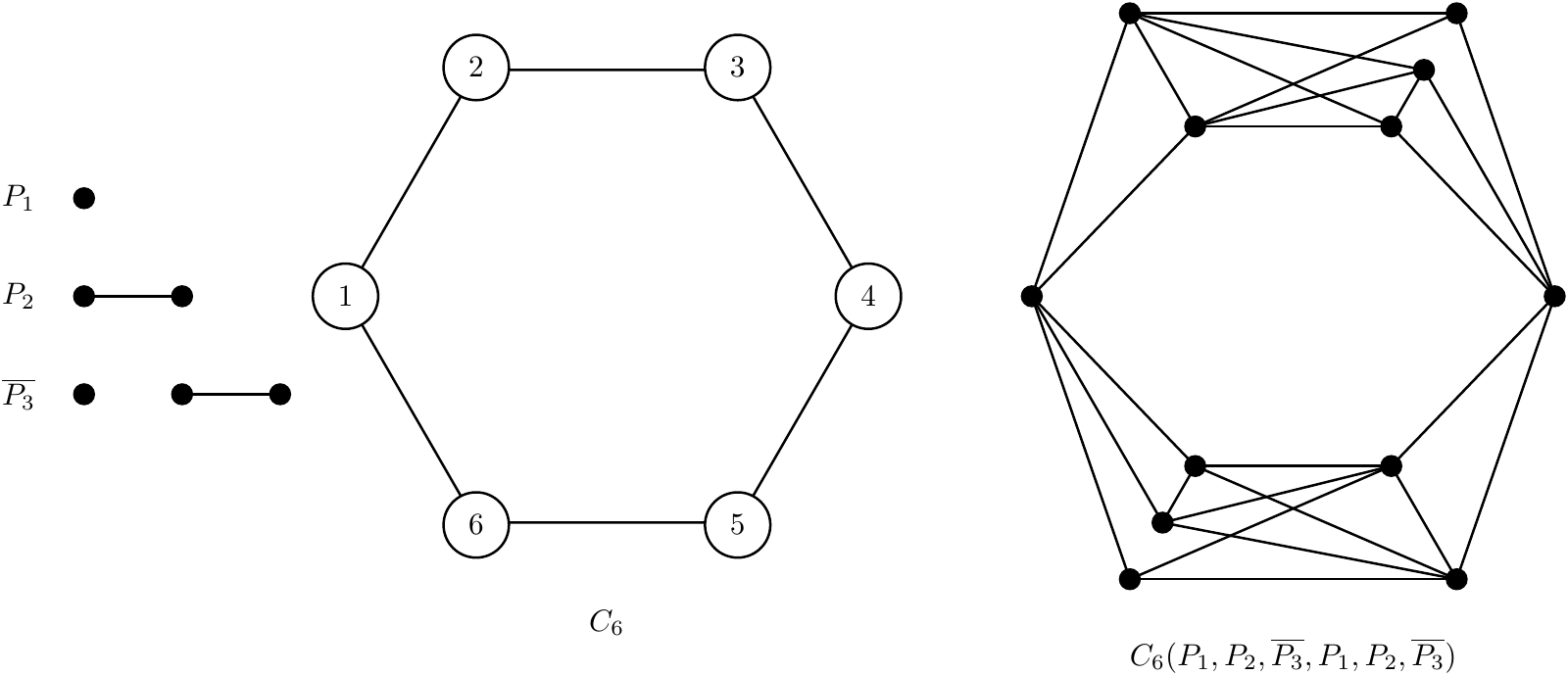}
    \caption{Illustration of the composition operation.}
    \label{figure:composition}
\end{figure}

This construction includes for instance the disjoint union of two graphs: $G_1 \cup G_2 = \bar{K}_2(G_1, G_2)$ and the join of two graphs: $K_2(G_1, G_2)$. In the special case where only a single vertex $b$ is replaced by the graph $G_b$, we use the notation
\[
    H( b, G_b ) = H( K_1, \dots, K_1, G_b, K_1, \dots, K_1 ).
\]
This special case includes adding an open twin (a new vertex $b'$ with $\Neighbourhood( b' ) = \Neighbourhood( b )$): $H( b, \bar{K_2})$ and adding a closed twin ($\Neighbourhood[ b' ] = \Neighbourhood[ b ]$): $H( b, K_2 )$. In fact, any composition can be obtained by repeatedly replacing a single vertex.

\begin{lemma} \label{lemma:composition}
Let $G = H(G_1, \dots, G_n)$ be a graph composition, where $V(H) = \{1, \dots, n\}$. For all $0 \le i \le n$, let $G^i$ be defined as $G^0 = H$ and $G^i = G^{i-1}( i, G_i )$. Then $G^n = G$.
\end{lemma}

\begin{proof}
The proof is by induction on $0 \le i \le n$. Let $I = \{1, \dots, i\}$ and $J = \{i+1, \dots, n\}$. We prove that
\begin{align*}
    V( G^i ) &= J \cup \{ v_a : a \in I, v \in V( G_a ) \}, \\
    E( G^i ) &= E( H[J] ) \cup \{ v_a v'_a : a \in I, vv' \in E( G_a ) \} \\
    & \ \cup \{ v_a v'_b : ab \in E( H[I] ) \} \cup \{ h v_a : h \in J, a \in I, ha \in E( H ) \}.
\end{align*}

This is clear for $i=0$, and for $i \ge 1$ this is easily verified from the recurrence property:
\begin{align*}
    V( G^i ) &= ( V( G^{i-1} ) \setminus \{ i \} ) \cup \{ v_i : v \in V( G_i ) \} \\
    E( G^i ) &= E( G^{i-1} - i ) \cup \{ v_i v'_i : vv' \in E( G_i ) \} \cup \{ h v_i : v \in V( G_i ), hi \in E( H ) \}.
\end{align*}
The details are omitted.
\end{proof}

\begin{proposition}
If $G = H(G_1, \dots, G_n)$, where each of $H, G_1, \dots, G_n$ is permissible, then $G$ is permissible.
\end{proposition}

\begin{proof}
According to Lemma \ref{lemma:composition}, we only need to prove the case where $G = H(b, G_b)$, where the vertices are sorted according to a permis $\hat{w} = \hat{w}_{1 : b-1} \hat{w}_b \hat{w}_{b+1 : n}$ of $H$. We denote the vertex set of $G_b$ as $V_b$, and we let $w^b$ be a permis of $G_b$. Then we claim that $w = \hat{w}_{1 : b-1} w^b \hat{w}_{b+1 : n}$ is a permis of $G$. 

For any configuration $x$ of $G$, let $\hat{x}$ be the configuration of $H$ such that $\hat{x}_u = x_u$ for all $u \ne \hat{w}_b$ and $\hat{x}_{ \hat{w}_b} = \bigjoin_{v \in V_b} x_v$. We then prove that $y \in \MaximalIndependentSets( G )$ by considering the three main steps of $w$. 
    \begin{itemize}
        \item Step 1: before the update of $G_b$ ($\hat{w}_{1:b-1}$). \\
        In Step 1, the initial configuration is $x$ and the final configuration is $\alpha = y^{b-1}$. It is easy to show that for any $1 \le a < b$, we have $\MIS^{ w_{1 : a} }( x; G )_{ G - V_b } = \MIS^{ \hat{w}_{1 : a} }( \hat{x}; H )_{H - \hat{w}_b}$.  We obtain $\hat{ \alpha } = \MIS^{ \hat{w}_{1 : b-1} }( \hat{x}; H )$.

        \item Step 2: update of $G_b$ ($w^b$). \\
        In Step 2, the initial configuration is $\alpha = y^{b-1}$ and the final configuration is $\beta = y^{b-1 + |V_b|}$. Note that $V_b$ is a tethered set of $G$, so let $T = \Neighbourhood( V_b; G ) \setminus V_b$.  If $\alpha_{ T } \ne 0$, then the whole of $G_b$ will be updated to $0$: $\beta_{V_b} = 0$. Otherwise, it is as if $G_b$ is isolated from the rest of the graph and $\beta_{V_b} = \MIS^{w^b}( \alpha_{V_b} ; G_b )$. In either case, we have $\hat{ \beta } = \MIS^{ \hat{w}_b }( \hat{ \alpha }; H )$.

        \item Step 3: after the update of $G_b$ ($\hat{w}_{b+1:n}$). \\
        In Step 3, the initial configuration is $\beta = y^{b-1 + |V_b|}$ and the final configuration is $y$. Again, we have for all $b < a \le n$, $\MIS^{ w_{b+1 : a} }( \beta; G )_{ G - V_b } = \MIS^{ \hat{w}_{b+1 : a} }( \hat{ \beta }; H )_{H - \hat{w}_b}$. We obtain $\hat{y} = \MIS^{ \hat{w}_{b+1 : n} }( \hat{ \beta }; H )$.
    \end{itemize}

    We obtain
    \begin{align*}
        \hat{y} &= \MIS^{ \hat{w}_{b+1 : n} }( \hat{ \beta }; H ) \\
        &= \MIS^{ \hat{w}_{b+1 : n} }( \MIS^{ \hat{w}_b }( \hat{ \alpha }; H ) ; H ) \\
        &= \MIS^{ \hat{w}_{b+1 : n} }( \MIS^{ \hat{w}_b }( \MIS^{ \hat{w}_{1 : b-1} }( \hat{x}; H ) ; H ) ; H ) \\
        &= \MIS^{ \hat{w} }( \hat{x}; H ).
    \end{align*}
    We can now prove that $y_{ \Neighbourhood[ v; G ] } \ne 0$ for every vertex $v$ of $G$. First, if $v$ is not a vertex of $G_b$, then $\hat{y}_{ \Neighbourhood[ v; H ] } \ne 0$, and hence $y_{ \Neighbourhood[ v; G ] } \ne 0$. Second, if $v = u_b$ is a vertex of $G_b$, then  we need to consider two cases. Either $\hat{y}_b = 0$, in which case there exists $a \in \Neighbourhood( b; H ) \subseteq \Neighbourhood( u_b ;G )$ with $\hat{y}_a = y_a \ne 0$; or $\hat{y}_b = 1$, in which case $y_{V_b} = \MIS^{w^b}( x_{ V_b }; G_b ) \in \MIS( G_b )$ and in particular $y_{ \Neighbourhood[ u_b ;G ] } \ne 0$.
\end{proof}

\subsection{Non-permissible graphs}

We now exhibit classes of non-permissible graphs. As mentioned earlier, the smallest non-permissible graph is the heptagon; in fact, any odd hole with at least seven vertices is non-permissible.

\begin{proposition} \label{proposition:bad_cycles}
For all $k \ge 3$, the odd hole $C_{2k+1}$ is not permissible.
\end{proposition}

\begin{proof}
Let $w$ be a permutation of the vertex set of $G = C_{2k+1}$. We shall prove that if $w$ is a permis there cannot be two consecutive arcs in $G^w$ with the same direction; this shows that the direction of arcs must alternate, which is impossible because there is an odd number of arcs in the cycle. We do this by a case analysis on the arcs preceding those two consecutive arcs.


We consider six vertices $a$ to $f$, where the first two arcs $a \gets b \gets c$ are in the same direction. The first case is where $c \gets d$; in that case, if $x_{abc} = 111$, then $y_{bcd} = 000$ and hence $c$ is not covered. This is shown in Figure \ref{figure:bad_cycles}, along with the other three cases.

\begin{figure}
    \centering
    \includegraphics[page=14,width=.5\textwidth]{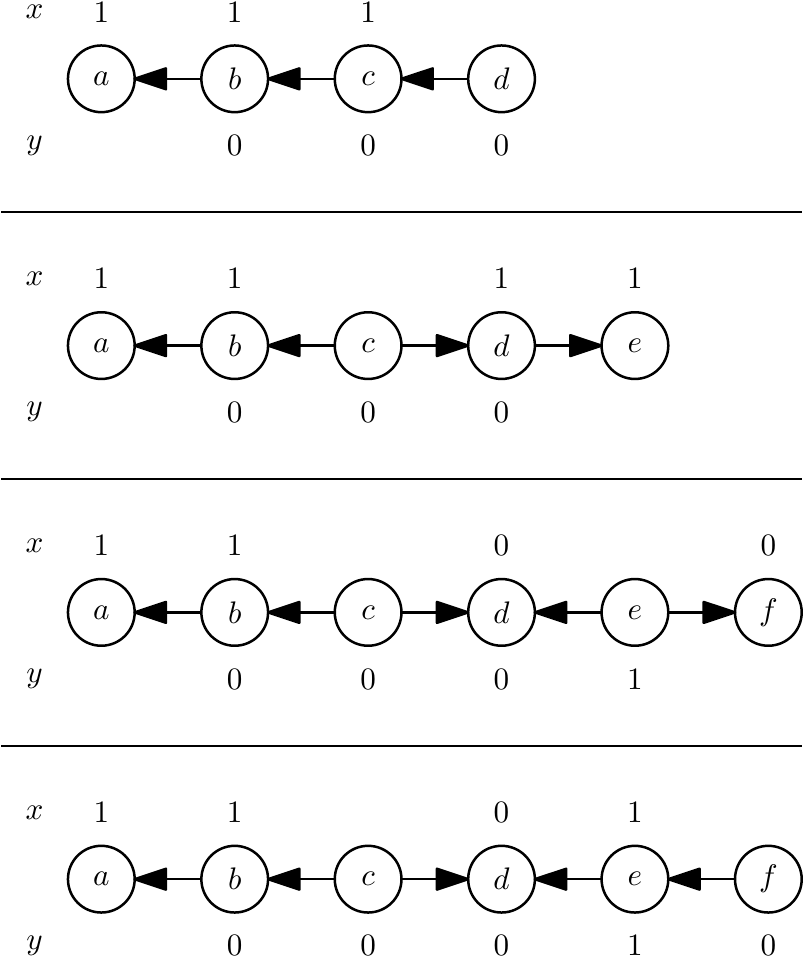}

    \caption{Illustration of Proposition \ref{proposition:bad_cycles}.}

    \label{figure:bad_cycles}
\end{figure}

\end{proof}


We now give two ways to construct larger non-permissible graphs.

Recall that a set of vertices $S$ is tethered if there is an  edge $st$ between any $s \in S$ and any $t \in T = \Neighbourhood(S) \setminus S$.

\begin{proposition} \label{proposition:blocking_set}
Let $G$ be a graph. If $G$ has a tethered set of vertices $S$ such that $G[ S ]$ has no permis, then $G$ has no permis.
\end{proposition}

\begin{proof}
Let $w$ be any permutation of $V$ and $\hat{w} = w[S]$. Let $\hat{x}$ be a configuration of $G[S]$ which is not fixed by $\hat{w}$: $\MIS^{\hat{w}} ( \hat{x}; G[S] ) \notin \MaximalIndependentSets( G[S] )$. We first note that $\hat{x} \ne 0$ and that for all $0 \le a \le |\hat{w}|$, $\MIS^{ \hat{w}_{1 : a} }( \hat{x}; G[S] ) \ne 0$.

Let $T = \Neighbourhood( S ) \setminus S$ and $U = V \setminus ( S \cup T )$ and $x = ( x_S = \hat{x}, x_T = 0, x_U )$, where $x_U$ is any partial configuration. We prove by induction on $0 \le b \le |w|$ that 
\[
    y^b := \MIS^{ w_{1 : b} }( x; G ) = \left( y^b_S = \MIS^{ \hat{w}_{1 : b'} }( \hat{x}; G[S] ) , y^b_T = 0, y^b_U \right),
\]
where $b'$ is defined by $[ \hat{w}_{1 : b'} ] = S \cap [ w_{1 : b} ]$. The base case $b = 0$ is clear. Suppose it holds for $b - 1$.
\begin{itemize}
    \item Case 1: $w_b \in S$. \\
    Then $b' = (b-1)' + 1$ and $w_b = \hat{w}_{ b' }$. Since $y^{b-1}_T = 0$, we have 
    \[
        y^b_{w_b} = \MIS( y^{b-1} ; G )_{w_b} = \MIS( y^{b-1}_S ; G[S] )_{w_b} = \MIS(  \MIS^{ \hat{w}_{1 : b' - 1} }( \hat{x} ; G[S] ) ; G[S]  )_{ \hat{w}_{ b' } } = \MIS^{ \hat{w}_{1 : b'} }( \hat{x}; G[S] )_{w_b},
    \]
    and hence $y^b_S = \MIS^{ \hat{w}_{1 : b'} }( \hat{x}; G[S] )$.

    \item Case 2: $w_b \in T$. \\
    Then $b' = (b-1)'$. Since $y^{b-1}_S \ne 0$, we have $\MIS( y^{b-1} ; G )_{w_b} = 0$ and hence $y^b_T = 0$.

    \item Case 3: $w_b \in U$. \\
    This case is trivial.
\end{itemize}

For $b = |w|$ we obtain $y = \MIS^w(x ; G) = ( \MIS^{ \hat{w} } ( \hat{x}; G[S] ), 0, y_U )$, for which $y_S \notin \MaximalIndependentSets( G[ S ] )$, and hence $y \notin \MaximalIndependentSets( G )$.
\end{proof}

Propositions \ref{proposition:bad_cycles} and \ref{proposition:blocking_set} yield perhaps the second simplest class of non-permissible graphs. The \Define{wheel graph} is $W_{n+1} = K_2( C_n, K_1 )$.

\begin{corollary} \label{corollary:bad_wheels}
For all $k \ge 3$, the wheel graph $W_{2k+2}$ is not permissible.
\end{corollary}

Clearly, a graph is permissible if and only if all its connected components are permissible. In particular, for any $G$, the union $H = G \cup C_7$ is not permissible, but is disconnected. An interesting consequence of Proposition \ref{proposition:blocking_set} is that permissibility of a connected graph cannot be decided by focusing on an induced subgraph, even if the latter has all but seven vertices of the original graph. Indeed, for any graph $G$, the join $H' = K_2( C_7, G )$ is not permissible, since the heptagon is tethered in $H'$.

\begin{corollary} \label{corollary:non-permissible_induced}
Let $G$ be a graph. Then there exists a connected non-permissible graph $H'$ such that $G$ is an induced subgraph of $H'$.
\end{corollary}

Second, and unsurprisingly, we can construct larger non-permissible graphs by using a constituency.

\begin{proposition} \label{proposition:no_permis_constituency}
Let $G$ be a graph. If there exists $A \subseteq V$ such that $G[A]$ is not permissible and $S = \Neighbourhood( A ) \setminus A$ is a constituency of $G - A$, then $G$ is not permissible.
\end{proposition}

\begin{proof}
Let $I$ be an independent set of $G - A$ such that $S \subseteq \Neighbourhood( I )$. Let $w$ be a permutation of $V$ and $w[A]$ be the permutation of $A$ induced by $w$. Then let $\dot{ x } \in \{0,1\}^A$ be a configuration such that $\dot{ y } = \MIS^{ w[A] }( \dot{x} ) \notin \MaximalIndependentSets( G[A] )$. Then consider $x \in \{0,1\}^V$ such that: $x_I = 1$, $x_A = \dot{ x }$, $x_{V - A - I} = 0$. We then have $y_I = 1$, $y_S = 0$, and hence $y_A = \dot{ y }$. Since $\dot{ y } \notin \MaximalIndependentSets( G[A] )$ and $y_S = 0$, we obtain that $y \notin \MaximalIndependentSets( G )$, i.e. $w$ is not a permis of $G$.
\end{proof}

For all $k \ge 3$ the odd hole $C_{2k+1}$ is non-permissible. Consider the graph $C_{2k+1+}$ by adding a vertex of degree one to $C_{2k+1}$; as we shall see later it is permissible. Now add another vertex of degree one to the tail of $C_{2k+1+}$ to obtain $C_{2k+1++}$. The graphs $C_7$, $C_{7+}$ and $C_{7++}$ are illustrated in Figure \ref{figure:c7_c7+_c7++}. In $C_{7++}$, the vertex $S = \{ \eta \}$ is a constituency and is the neighbourhood of the heptagon; therefore $C_{7++}$ is not permissible. Obviously, this reasoning applies to all larger $C_{2k+1++}$ as well.

\begin{corollary} \label{corollary:odd_cycle++}
For all $k \ge 3$, $C_{2k+1++}$ is not permissible.
\end{corollary}

\begin{figure}
    \centering
    \includegraphics[page=13,width=\textwidth]{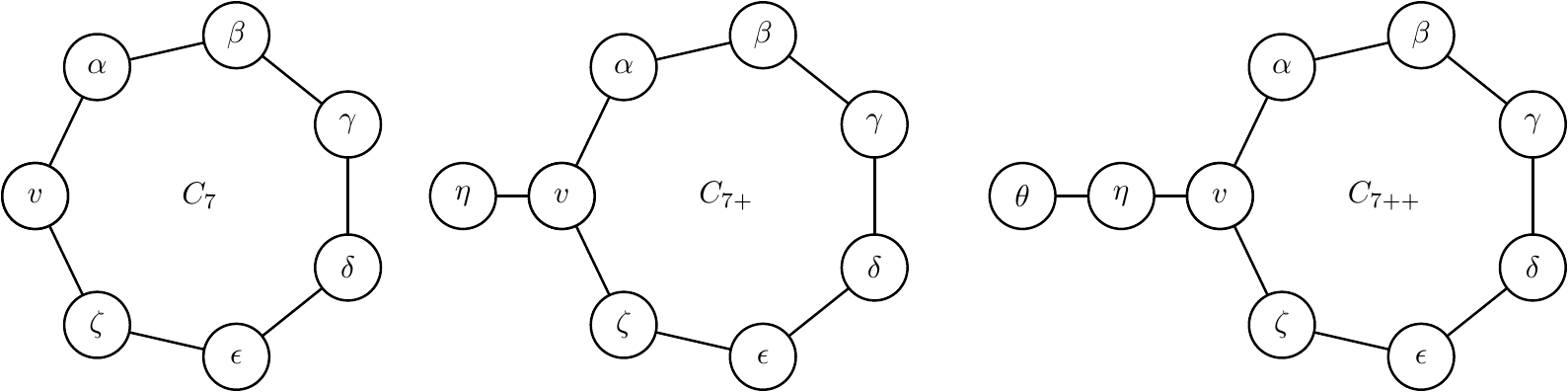}
    \caption{Graphs $C_{7}$ (non-permissible), $C_{7+}$ (permissible), and $C_{7++}$ (non-permissible).}
    \label{figure:c7_c7+_c7++}
\end{figure}

\subsection{The \DecisionProblem{Permissible} decision problem}

We now prove that deciding whether a graph is permissible is computationally hard.

\begin{decisionproblem}
    \problemtitle{Permissible}
    \probleminput{A graph $G$.}
    \problemquestion{Is $G$ permissible?}
\end{decisionproblem}

\begin{figure}
    \centering
    \includegraphics[page=12,width=0.5\textwidth]{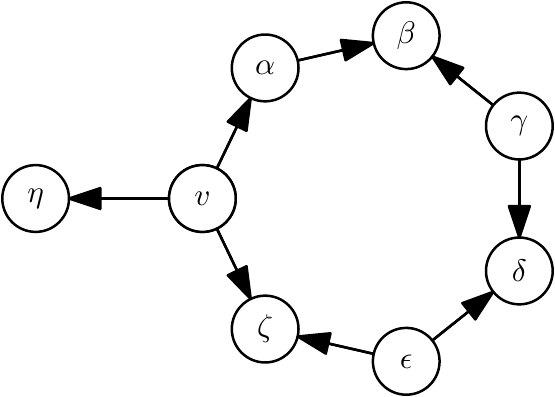}
    \caption{The permis of $C_{7+}$.}
    \label{figure:permis_C7+}
\end{figure}

As mentioned above, $C_{7+}$ is permissible; the proof of Lemma \ref{lemma:permis_C7+} below can easily be generalised to show that $C_{2k+1+}$ is permissible for all $k \ge 3$.

\begin{lemma} \label{lemma:permis_C7+}
Any $w$ such that ${ C_{7+} }^w$ is given on Figure \ref{figure:permis_C7+} is a permis of $C_{7+}$.
\end{lemma}

\begin{proof}
The vertices $\beta$, $\delta$, $\zeta$, and $\eta$ are all transitive, and $v$ is near-transitive, and hence these are covered. We only need to show that $\alpha$, $\gamma$ and $\epsilon$ are covered.

For $\alpha$, suppose $y_{v \alpha \beta} = 000$. We then have the following chain of implications:
\begin{alignat*}{3}
    ( y_\alpha = 0 \meet y_\beta = 0 ) & \; \meet \; &  ( \alpha \to \beta \gets \gamma ) & \; \implies y_\gamma = 1 \\
    ( y_\gamma = 1 ) & \; \meet \; & ( \gamma \to \beta ) & \; \implies x_\beta = 0 \\
    ( x_\beta = 0 \meet  y_v = 0 ) & \; \meet \; & ( v \to \alpha \to \beta ) & \; \implies y_\alpha = 1,
\end{alignat*}
which is the desired contradiction.

For $\gamma$, suppose $y_{\beta \gamma \delta } = 000$. We then have the following chain of implications:
\begin{alignat*}{3}
    ( y_\beta = 0 \meet y_\gamma = 0 ) & \; \meet \; & ( \alpha \to \beta \gets \gamma ) & \; \implies y_\alpha = 1 \\
    ( y_\alpha = 1 ) & \; \meet \; & ( \alpha \to \beta ) & \; \implies x_\beta = 0 \\
    ( y_\gamma = 0 \meet y_\delta = 0 ) & \; \meet \; & ( \gamma \to \delta \gets \epsilon ) & \; \implies y_\epsilon = 1 \\
    ( y_\epsilon = 1 ) & \; \meet \; & ( \delta \gets \epsilon ) & \; \implies x_\delta = 0 \\
    ( x_\beta = 0 \meet x_\delta = 0 ) & \; \meet \; & ( \beta \gets \gamma \to \delta ) & \; \implies y_\gamma = 1,
\end{alignat*}
which is the desired contradiction.

For $\epsilon$, suppose $y_{\delta \epsilon \zeta} = 000$. We then have the following chain of implications:
\begin{alignat*}{3}
    ( y_\delta = 0 \meet y_\epsilon = 0 ) & \; \meet \; & ( \gamma \to \delta \gets \epsilon ) & \; \implies y_\gamma = 1 \\
    ( y_\gamma = 1 ) & \; \meet \; & ( \gamma \to \delta ) & \; \implies x_\delta = 0 \\
    ( y_\epsilon = 0 \meet y_\zeta = 0 ) & \; \meet \; & ( \epsilon \to \zeta \gets v ) & \; \implies y_v = 1 \\
    ( y_v = 1 ) & \; \meet \; & ( \zeta \gets v ) & \; \implies x_\zeta = 0 \\
    ( x_\delta = 0 \meet x_\zeta = 0 ) & \; \meet \; & ( \delta \gets \epsilon \to \zeta ) & \; \implies y_\epsilon = 1,
\end{alignat*}
which is the desired contradiction.
%




    

\end{proof}

\begin{theorem} \label{theorem:complexity_permissible}
    \DecisionProblem{Permissible} is \coNP-hard.
\end{theorem}

\begin{proof}
Reduction from \DecisionProblem{Non-Constituency}. Let $(G, S)$ be an instance of \DecisionProblem{Non-Constituency} of empty type, and let $T = V \setminus S$. Let $T'$ be a copy of $T$. Let $A$ be another set of $7$ vertices inducing a heptagon and let $v \in A$. Then consider the instance $\hat{G} = ( \hat{V}, \hat{E} )$ of \DecisionProblem{Permissible} with 
\begin{align*}
    \hat{V} &= V \cup T' \cup A \\
    \hat{E} &= E \cup \{ tt' : t \in T \} \cup C_7(A) \cup \{ sv : s \in S \}.
\end{align*}
This construction is illustrated in the left hand side of Figure \ref{figure:complexity_permissible}.

We claim that $(G, S)$ is a yes-instance of \DecisionProblem{Non-Constituency} if and only if $\hat{G}$ is a yes-instance of \DecisionProblem{Permissible}. First, if $S$ is a constituency of $G$, then $S$ is a constituency of $\hat{G} - A$, and hence $\hat{G}$ is not permissible by Proposition \ref{proposition:no_permis_constituency}. Second, if $S$ is not a constituency of $G$, let $\omega$ be the permis of $C_{7+}$ given in Lemma \ref{lemma:permis_C7+} and consider the orientation $\hat{G}^w$ such that $T' \cup S$ are all sinks and $\hat{G}^w[A \cup \{ s \}] = { C_{7+} }^\omega$ for all $s \in S$. (The edges of $\hat{G}[ T ]$ can be oriented arbitrarily.) This is illustrated in the right hand side of Figure \ref{figure:complexity_permissible}.

The vertices in $T \cup T' \cup S$ are all near-transitive, and hence covered by $w$. We now prove that the vertices in $B = A \setminus \{ v \}$ are also covered by $w$. Let $\dot{ x } \in \{0,1\}^{ A \cup \{ \eta \} }$ be given by $\dot{ x }_a = x_a$ for all $a \in A$ and $\dot{ x }_\eta = \bigmeet_{s \in S} \neg x_s$, and let $\dot{ y } \in \{ 0,1 \}^{A \cup \{ \eta \} } = \MIS^\omega( \dot{ x }; C_{7+} )$. Since $A \succ V \setminus A$ in $\hat{G}^w$, we obtain $y_A = \MIS^w( x; \hat{G} )_A = \MIS^\omega( \dot{x}; C_{7+} )_A = \dot{ y }_A$. For all $u \in B$, we have $\Neighbourhood[ u; \hat{G} ] = \Neighbourhood[ u; C_{7+} ]$ and since $\omega$ is a permis, $y_{ \Neighbourhood[ u; \hat{G} ] } = \dot{ y }_{ \Neighbourhood[ u; C_{7+} ] } \ne 0$, thus $u$ is covered by $w$. Therefore, the only vertex in contention is $v$. Suppose $y_v = 0$ and let $I = \one( y ) \cap T$. Since $I$ is an independent set of $G$ and $S$ is not a constituency of $G$, there exists $s \in S$ outside of the neighbourhood of $I$. Thus, $y_{ \Neighbourhood(s) } = 0$ and hence $y_s = 1$, which yields $y_{\Neighbourhood[ v ]} \ne 0$. Thus $v$ is covered, and $w$ is a permis of $\hat{G}$.
\end{proof}

\begin{figure}
    \centering
    \includegraphics[page=7,width=\textwidth]{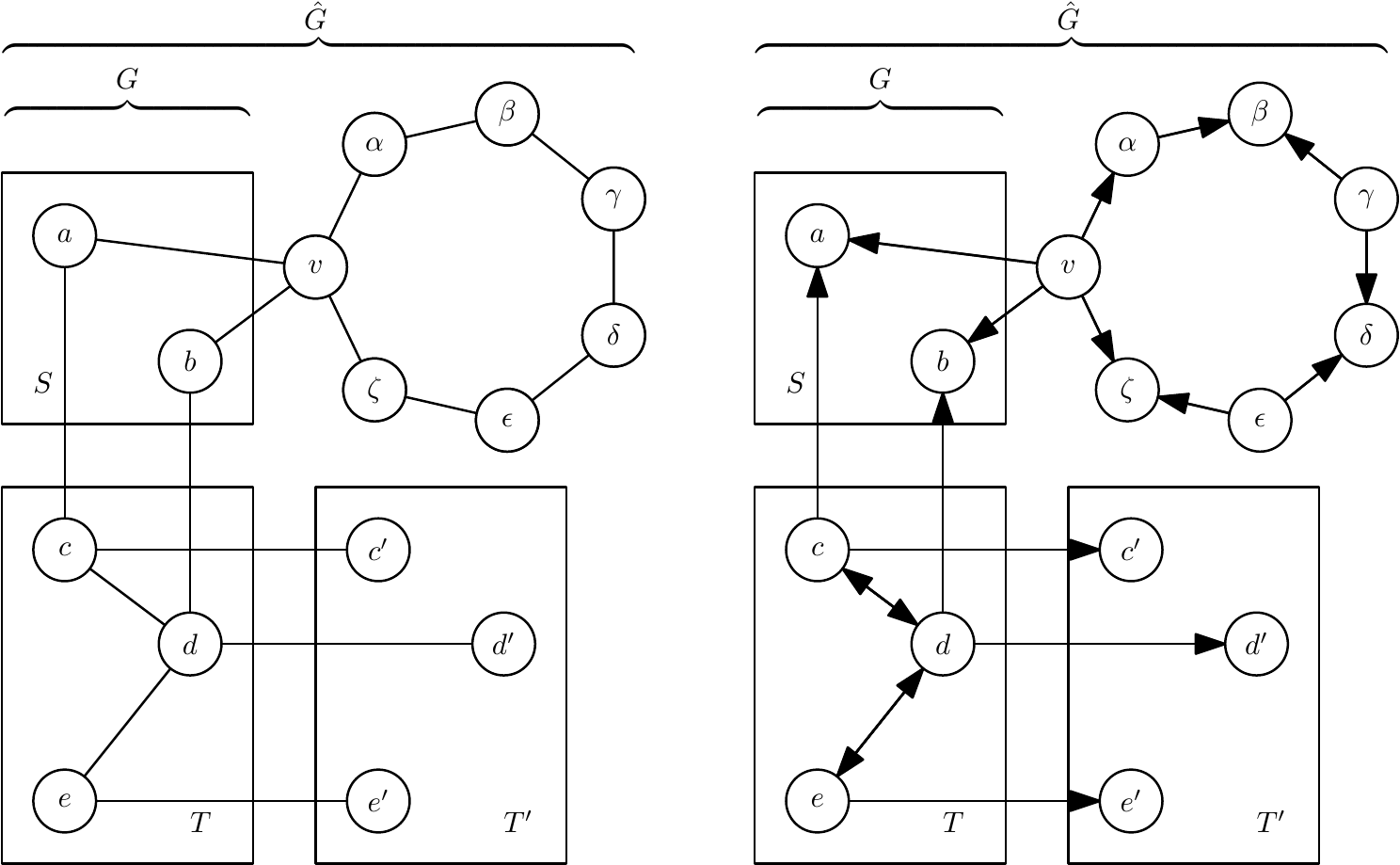}
    \caption{Illustration of the reduction from \DecisionProblem{Non-Constituency} to \DecisionProblem{Permissible}.}
    \label{figure:complexity_permissible}
\end{figure}

\section{Extension to kernels in digraphs} \label{section:digraphs}

The MIS network can be easily extended to digraphs as follows. The \Define{kernel network}, denoted by $\Kernel(G)$, is defined by 
\[
    \Kernel(x)_v = \bigmeet_{u \to v} \neg x_u,
\]
where $\Kernel(x)_v = 1$ if $\InNeighbourhood(v) = \emptyset$. We then have $\Fix( \Kernel( G ) ) = \Kernels(G)$.

\subsection{Fixability of the kernel network} \label{subsection:sequential_convergence_K}


Henceforth we consider the fixability of the kernel network.  We saw that $\vec{C}_{2k+1}$ has no kernel, and hence it is not fixable. However, the existence of kernels does not guarantee fixability, as we shall prove in Proposition \ref{proposition:no_convergence_K}.

First, we give some sufficient conditions for a digraph to have a kernel and yet to not be fixable. Generalising the definition from the undirected case, we say a set of vertices $S$ is \Define{tethered} if there is an undirected edge $st$ between any $s \in S$ and any $t \in T = \InNeighbourhood(S) \setminus S$. Note that necessarily $T \subseteq \OutNeighbourhood(S)$ but possibly $\OutNeighbourhood(S) \setminus T\neq \emptyset$.

\begin{proposition} \label{proposition:no_convergence_K}
If $G$ has a tethered set $S$ such that $\Kernel( G[ S ] )$ is not fixable, then $\Kernel( G )$ is not fixable.
\end{proposition}

We need two lemmas before giving the proof of Proposition \ref{proposition:no_convergence_K}.


\begin{lemma} \label{lemma:zero_not_reachable_K}
If $x \in \{0,1\}^V \ne 0$ and $x \mapsto_\Kernel y$, then $y \ne 0$.
\end{lemma}

\begin{proof}
Suppose, for the sake of contradiction, that $y = \Kernel^w(x) = 0$ with $w = w_{1 : l}$, while $y^{l-1} \ne 0$. Then $y^{l-1}_{\InNeighbourhood( w_l )} = 0$, hence $y_{w_l} = \Kernel( y^{l-1} )_{w_l} = 1$, which is the desired contradiction.
\end{proof}

\begin{lemma} \label{lemma:nonzero_not_fixable}
    If $\Kernel( G )$ is not fixable, then for any word $w$, there exists a \emph{nonzero} configuration $x$ of $G$ such that $\Kernel^w( x ) \notin \Kernels( G )$.
\end{lemma}

\begin{proof}
For the sake of contradiction, suppose $\Kernel^w( x ) \in \Kernels( G )$ for all nonzero $x$. We remark that $\Kernel( 0 )_{ w_1 } = 1$ hence $\Kernel^{ w_1 }( 0 ) \ne 0$ and by Lemma \ref{lemma:zero_not_reachable_K}, $\Kernel^w( 0 ) \ne 0$. Therefore, $\Kernel^{ww}( 0 ) \in \Kernels( G )$, and by hypothesis $\Kernel^{ww}( x ) \in \Kernels( G )$ for any nonzero $x$, which shows that $ww$ fixes $\Kernel( G )$.
\end{proof}

\begin{proof}[Proof of Proposition \ref{proposition:no_convergence_K}]
Partition the vertex set of $G$ into three parts: $S$, $T = \InNeighbourhood(S) \setminus S$, and $U = V \setminus (S \cup T)$.

Let $w \in V^*$ be a word, then by Lemma \ref{lemma:nonzero_not_fixable} there exists a nonzero configuration $\hat{x}$ of $G[S]$ such that $\Kernel^{ w[ S ] }( \hat{x}; G[S] ) \notin \Kernels( G[S] )$. Let $x$ be a configuration of $G$ such that $x_S = \hat{x} \ne 0$ and $x_T = 0$. By induction on $0 \le a \le l$, we prove that $y^a_S \ne 0$ and $y^a_T = 0$. The claim is clear for $a = 0$, hence suppose it holds for $a - 1$. We consider three cases.
\begin{itemize}
    \item Case 1: $w_a \in S$. \\
    Since $y^{a-1}_T = 0$, we have $y^a_S = \Kernel^{ w_a }( y^{a-1}; G )_S = \Kernel^{ w_a }( y^{a-1}_S; G[ S ] ) \ne 0$ by Lemma \ref{lemma:zero_not_reachable_K}. Also, $y^a_T = y^{a-1}_T = 0$.

    \item Case 2: $w_a \in T$. \\
    Since $y^{a-1}_S \ne 0$, we have $\Kernel( y^{a-1} ; G )_{w_a} = 0$. Also, $y^a_S = y^{a-1}_S \ne 0$.

    \item Case 3: $w_a \in U$. \\
    Then $y^a_S = y^{a-1}_S \ne 0$ and $y^a_T = y^{a-1}_T = 0$.
\end{itemize}

Therefore $y_S = \Kernel^{ w[ S ] }( x_S; G[S] )$ and hence $y_S \notin \Kernels( G[ S ] )$. Since $y_T = 0$, this implies that $y \notin \Kernels( G )$.
\end{proof}

In particular, if $G$ has a tethered set $S$ with $G[ S ] = \vec{C}_{2k+1}$ and $G - S$ is a graph, then $G$ has a kernel (namely every maximal independent set of $G - S$) but $\Kernel( G )$ is not fixable.




We are now interested in the computational complexity of deciding whether the kernel network is fixable.

\begin{decisionproblem}
    \problemtitle{Fixable}
    \probleminput{A digraph $G$.}
    \problemquestion{Is $\Kernel( G )$ fixable?}
\end{decisionproblem}

\begin{theorem} \label{theorem:complexity_fixable}
\DecisionProblem{Fixable} is \coNP-hard.
\end{theorem}

\begin{proof}
The proof is by reduction from \DecisionProblem{Tautology}, which is \coNP-hard. Let $\phi$ be a DNF with set of variables $\Variables$, set of literals $\Literals = \{ \Variable, \neg \Variable : \Variable \in \Variables \}$, and set of clauses $\Clauses$, so that $\phi$ can be expressed as $\phi = \bigjoin_{\Clause \in \Clauses} \bigmeet_{\Literal \in \Literals_\Clause} \Literal$ with $\Literals_\Clause \subseteq \Literals$ for all $\Clause \in \Clauses$. We construct a graph that has a vertex per literal, then builds a NOR-gate circuit for $\phi$, and finally attaches the output of this circuit to a $\vec{C}_3$. Intuitively, a $0$ output of the circuit (which occurs whenever $\phi$ is not a tautology) will isolate the $\vec{C}_3$, which is not fixable, while a $1$ output will guarantee convergence.

Let $G = (V, E)$ with
\begin{align*}
    V &= \Literals \cup \Clauses \cup \{ \neg \phi, \phi, a, b, c \}, \\
    E &= \{ \Variable \neg \Variable: \Variable \in \Variables \}
    \cup \{ \Literal \to \Clause : \neg \Literal \in \Clause, \Clause \in \Clauses \} 
    \cup \{ \Clause \to \neg \phi : \Clause \in \Clauses \} 
    \cup \{ \neg \phi \to \phi \to a \to b \to c \to a \}.
\end{align*}

This is illustrated in Figure \ref{figure:complexity_fixable}.

If $\phi$ is a tautology, then let $w = w^\Literals w^\Clauses \neg \phi \phi abc$, where $w^\Literals$ and $w^\Clauses$ are any permutations of $\Literals$ and $\Clauses$, respectively. For each $\Variable$, let $\Variable' = \{ \Variable, \neg \Variable \}$. Then $G[ \Variable' ]$ is a complete graph, and has $w[ \Variable' ]$ as a permis. 
Therefore, at the end of $w^\Literals$, we have $y_{ \Variable' } =  y^{ |\Literals| }_{ \Variable'} \in \{01, 10\}$. Let $z \in \{ 0,1 \}^\Variables = y^{ |\Literals| }_\Variables$. 
The graph $G$ then induces a circuit for $\phi$, so that $y_\Clause = y^{ |\Literals| + |\Clauses| }_\Clause = \bigmeet_{ \Literal \in \Clause } \neg y^{ |\Literals| }_{\neg \Literal} = \Clause( z )$ and $y_{ 
\neg \phi } = \neg \phi(z)$ and eventually $y_\phi = \phi(z) = 1$. 
Finally, we obtain $y_{abc} = 010$, and it is easily verified that $y \in \Kernels( G )$.

If $\phi$ is not a tautology, let $w$ be any word and construct the configuration $x$ as follows. Let $z \in \{ 0,1 \}^\Variables$ such that $\phi( z ) = 0$ then let $x_\Variable = z_\Variable$, $x_{ \neg \Variable } = \neg z_\Variable$, $x_\Clause = \Clause( z )$, $x_{ \neg \phi } = \neg \phi(z)$, $x_\phi = \phi( z ) = 0$ and $x_{abc}$ can be chosen arbitrarily, since in a $\vec{C_3}$ every configuration is such that $\Kernel^{ w[abc] }( x_{abc} ) \notin \Kernels( G[abc] )$. 
By construction, updating the value of any vertex outside of the $\vec{C}_3$ will not change anything, i.e. $y = \Kernel^{ w[abc] }( x; G )$. Now, since $x_\phi = y_\phi = 0$, we have $\Kernel^{ w[abc] }( x; G )_{ abc } = \Kernel^{ w[abc] }( x; G[abc] ) \notin \Kernels( G[abc] )$, and hence $y \notin \Kernels( G )$.
\end{proof}

\begin{figure}
    \centering
    \includegraphics[page=4,width=\textwidth]{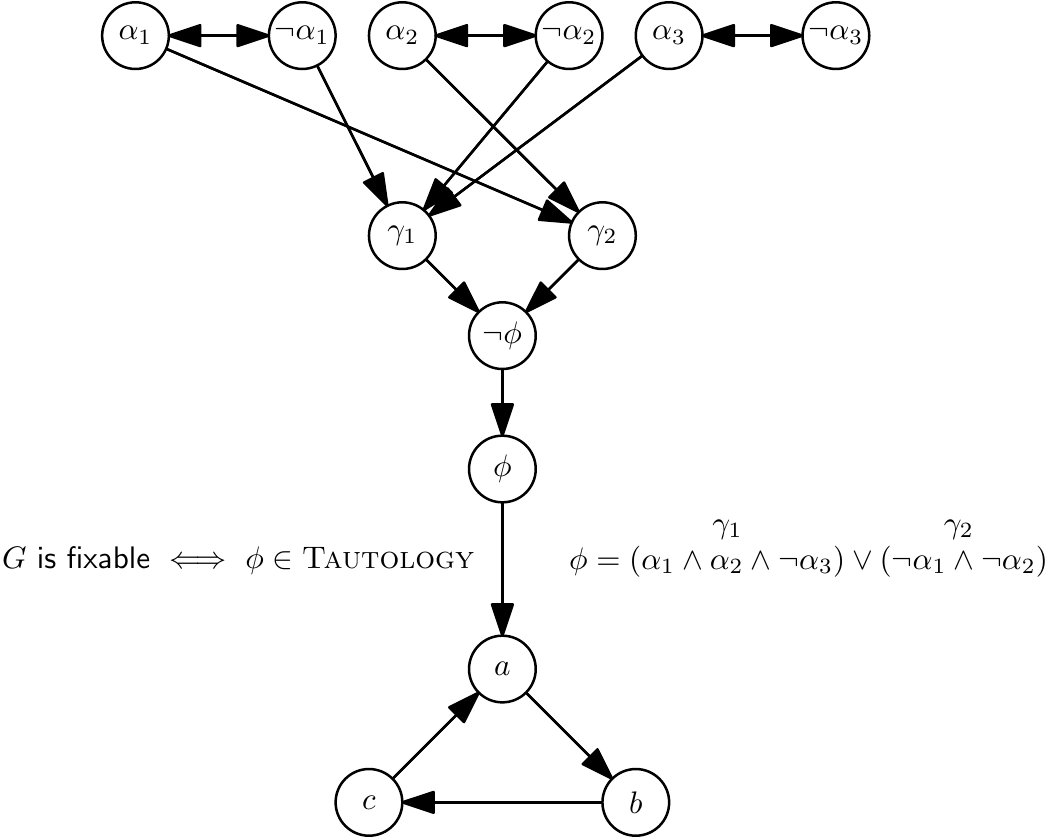}
    \caption{Illustration of the reduction from \DecisionProblem{Tautology} to \DecisionProblem{Fixable}.}
    \label{figure:complexity_fixable}
\end{figure}

We now refine Theorem \ref{theorem:complexity_fixable} by restricting ourselves to a very specific class of digraphs. For any $\epsilon > 0$, say a digraph $G$ is \Define{$\epsilon$-fixable} if there is a word $w$ that fixes at least a fraction $1 - \epsilon$ of configurations, i.e. $| \{ x \in \{ 0,1 \}^V : \Kernel^w( x ) \in \Kernels( G ) \} | \ge ( 1 - \epsilon ) 2^n$.

\begin{theorem} \label{theorem:complexity_fixable_refined}
For any $\epsilon > 0$, \DecisionProblem{Fixable} is \coNP-hard for oriented $\epsilon$-fixable digraphs of maximum in- and out-degree $2$.
\end{theorem}

\begin{proof}
The proof is by reduction from \DecisionProblem{3-Tautology} restricted to expressions where each literal appears at most twice, which is \coNP-hard. We use the notation introduced in the proof of Theorem \ref{theorem:complexity_fixable}.

Let $\phi$ be a \DecisionProblem{3-Tautology} instance where each literal appears at most twice. Let $\OtherVariables$ be an additional set of $\lceil -\log_2 \epsilon \rceil$ variables, and let
\[
    \psi = \phi \join \bigjoin \OtherVariables.
\]
Let $G$ be the graph constructed in the proof of Theorem \ref{theorem:complexity_fixable} for the expression $\psi$. Now replace each $K_2$ corresponding to a variable by a $\vec{C}_4$ as follows. Let $u$ be a variable in $\Variables \cup \OtherVariables$ and let $u' = \{ u, \neg u, \dot{u}, \neg \dot{u} \}$ induce the cycle $u \to \neg u \to \dot{u} \to \neg \dot{u} \to u$. We finally adapt the NOR-gate circuit of $\psi$ to have fan-in $2$. The new graph $\hat{G}$ is illustrated in Figure \ref{figure:complexity_fixable_refined}.

By construction, $\hat{G}$ is oriented. Since each literal appears at most twice in $\psi$, the out-degree of each vertex in $\vec{C_4}$ is at most $2$; it is easily verified that the in-degree at out-degree of each vertex is then at most $2$. Also, as before, $\hat{G}$ is fixable if and only if $\psi$ is a tautology, which in turn occurs if and only if $\phi$ is a tautology. 

We now prove that the new graph $\hat{G}$ is $\epsilon$-fixable by exhibiting a word $w$ that fixes all configurations with $x_\OtherVariable = 1$ for some $\OtherVariable \in \OtherVariables$. Let $u \in \Variables \cup \OtherVariables$, then the word $w^u = \neg u \dot{u} \neg \dot{u}$ fixes $\hat{G}[ u' ]$ without updating the vertex $u$. Then let $w' = (w^u : u \in \Variables \cup \OtherVariables)$ be the concatenation of all the $w^u$ words (in any order). Let $w''$ be a word that follows the circuit in topological order, and finally let $w = w' w'' abc$. Let $x$ be a configuration such that $x_\OtherVariable = 1$ for some $\OtherVariable \in \OtherVariables$. Note that because we chose $|\OtherVariables|\ge -\log_2 \epsilon$ at least $1-\epsilon$ of all configurations satisfy this property. 
We have $\psi( x_{ \Variables \cup \OtherVariables } ) = \psi( y_{ \Variables \cup \OtherVariables } ) = 1$ and hence $y_\psi = 1$ and $y_{abc} = 010$. Thus $y \in \Kernels( \hat{G} )$.
\end{proof}

\begin{figure}
    \centering
    \includegraphics[page=8,width=\textwidth]{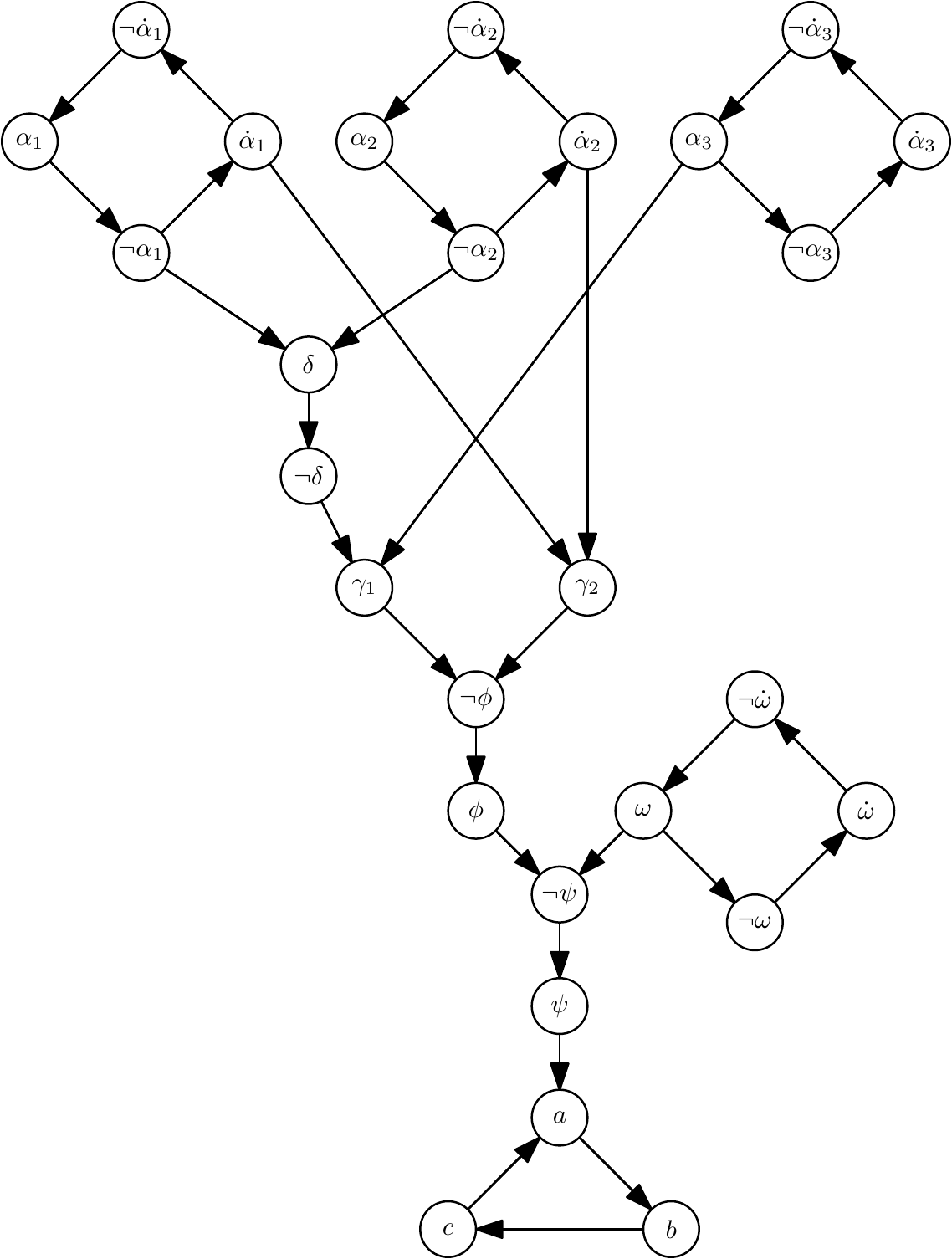}
    \caption{Illustration of the reduction from \DecisionProblem{3-Tautology} where each literal appears at most twice to \DecisionProblem{Fixable}.}
    \label{figure:complexity_fixable_refined}
\end{figure}

\subsection{The independent network} \label{subsection:convergence_I}

Since the kernel network is not fixable in general, and since it is not even tractable to decide whether the kernel network of a particular digraph is fixable, we now introduce and study two Boolean networks that are always fixable, and whose sets of fixed points contain all kernels.

The \Define{independent network} on $G$, denoted by $\Independent(G)$, is defined by
\[
    \Independent(x)_v = x_v \meet \bigmeet_{u \to v} \neg x_u,
\]
with $\Independent(x)_v = x_v$ if $\InNeighbourhood(v) = \emptyset$. We then have $\Fix( \Independent( G ) ) = \IndependentSets( G )$. Moreover, as we shall prove later, any permutation of $V$ is a fixing word of $\Independent(G)$.

We first settle the reachability problem for the independent network. We shall tacitly use the simnple fact that if $x \mapsto_\Independent y$, then $y \le x$. We characterise the configurations $y$ that are reachable from a given initial configuration $x$. Generalising the undirected case, for any configuration $x \in \{0,1\}^V$, we let $\mathcal{C}(x)$ denote the collection of initial strong components of $G[ \one( x ) ]$.

\begin{proposition}[Reachability for the independent network] \label{proposition:reachability_I}
Let $x, y \in \{0,1\}^V$ be two configurations. The following are equivalent:
\begin{enumerate}
    \item 
    $x \mapsto_\Independent y$;

    \item 
    $x \Geodesic_\Independent y$;

    \item 
    $y \le x$ and $y_C \ne 0$ for any $C \in \mathcal{C}(x)$. 
\end{enumerate}
\end{proposition}


\begin{proof}[Proof of Proposition \ref{proposition:reachability_I}]
If $x=0$ then the proposition trivially holds as $G[ \one( x ) ]$ is the empty digraph and $y=x$, so the geodesic is the empty word; we now consider $x\neq 0$.

We first prove that if $x \mapsto_\Independent y$, then $y_C \ne 0$ for all $C \in \mathcal{C}(x)$. Suppose $y = \Independent^w( x )$ satisfies $y_C = 0$, with $w = w_{1 :l}$, and that $y^{l-1}_C \ne 0$. Then $w_l \in C$ and there exists a vertex $u$ such that $y^{l-1}_u = x_u = 1$ and $u \to w_l$; but then $u \in C$ and hence $y_C \ne 0$, which is the desired contradiction.

We now prove that if $y \le x$ and $y_C \ne 0$ for all $C \in \mathcal{C}(x)$, then there is a geodesic from $x$ to $y$. We use Lemma \ref{lemma:out-rooted_forest}. The geodesic is a concatenation of words $w^B$, one for each strong component $B$ of $G[ \one( x ) ]$, in reverse topological order. The word $w^B$ is described as follows.
\begin{itemize}
    \item Case 1: $y_B = 0$. \\
    Then $B$ is not an initial component, hence there exists $u \in B$ with an in-neighbour in a parent strong component and $T$ a spanning out-tree of $B$ rooted at $u$. Then $w^B$ is the word obtained by traversing the tree from leaves to root.

    \item Case 2: $y_B \ne 0$. \\
    Let $S = \one( y ) \cap B$ and $T$ be the spanning out-forest of $B$ rooted at $S$. Then $w^B$ is the word obtained by traversing the out-forest but missing out the roots. (In particular, if $S = B$, then $w^B$ is empty.)
\end{itemize}
It is easy to verify that $w$ is indeed a geodesic from $x$ to $y$.
\end{proof}

We now determine the words that fix the independent network. We generalise the concept of a vertex cover to digraphs as follows: a \Define{directed vertex cover} of a digraph is a set of vertices $S$ such that for every symmetric edge $uv$, $\{ u, v \} \cap S \ne \emptyset$ and for every oriented edge $(u, v)$, $v \in S$.

\begin{proposition}[Words fixing the independent network] \label{proposition:fixing_words_I}
Let $G$ be a digraph. Then $w$ fixes the independent set network $\Independent(G)$ if and only if $[w]$ is a directed vertex cover of $G$.
\end{proposition}

\begin{proof}
Suppose $w$ fixes $\Independent$. If $uv$ is a symmetric edge with $\{u, v\} \cap [w] = \emptyset$, then for any configuration $x$ with $x_{uv} = 11$, we have $y_{uv} = 11$. If $(u, v)$ is an oriented edge with $v \notin [w]$, then if $x_{ \InNeighbourhood(u), u, v} = (0, 1, 1)$, we have $y_{uv} = 11$.

Suppose $[w]$ is a directed vertex cover. Suppose $y_{uv} = 11$. If $uv$ is symmetric, without loss let $v$ be updated last. Let $y^a$ be the configuration before the last update of $v$, then $y^a_u = y_u = 1$, which implies $y_v = 1$, which is a contradiction. If $(u, v)$ is oriented, then let $y^a$ be the configuration before the last update of $v$, then $y^a_u = y_u = 1$, which implies $y_v = 1$, which is a contradiction.
\end{proof}

\begin{decisionproblem}
  \problemtitle{$\Independent$ fixing word}
  \probleminput{A digraph $G = (V,E)$ and a word $w$.}
  \problemquestion{Does $w$ fix $\Independent(G)$?}
\end{decisionproblem}

\begin{corollary} \label{corollary:complexity_fixing_word_I}
\DecisionProblem{$\Independent$ fixing word} is in \P.
\end{corollary}

\subsection{The dominating network} \label{subsection:convergence_D}

The \Define{dominating network} on $G$, denoted by $\Dominating(G)$, is defined by
\[
    \Dominating(x)_v = x_v \join \bigmeet_{u \to v} \neg x_u,
\]
with $\Dominating(x)_v = 1$ if $\InNeighbourhood(v) = \emptyset$. We then have $\Fix( \Dominating( G ) ) = \DominatingSets( G )$. Moreover, once again, any permutation of $V$ fixes the dominating network. 

We first settle the reachability problem for the dominating network. We use the simple fact that if $x \mapsto_\Dominating y$, then $y \ge x$. We  characterise the configurations $y$ that are reachable from a given configuration $x$. For any configuration $x$, let $A(x) = \zero( x ) \cap \OutNeighbourhood( \one( x ) )$ and $B(x) = \zero( x ) \setminus \OutNeighbourhood( \one( x ) )$.

\begin{proposition}[Reachability for the dominating network] \label{proposition:reachability_D}
Let $x,y \in \{0,1\}^V$ be two configurations. The following are equivalent:
\begin{enumerate}
    \item 
    $x \mapsto_\Dominating y$;

    \item 
    $x \Geodesic_\Dominating y$;

    \item 
    $y \ge x$, $y_{A( x )} = 0$, and $G[ \one( y_{B( x )} ) ]$ is acyclic.
\end{enumerate}
\end{proposition}

\begin{proof}
We first prove that if $x \mapsto_\Dominating y$, then $y_{A( x )} = 0$. 
We prove that $y^a_{A( x )} = 0$ by induction on $a$. This is clear for $a = 0$, so assume it holds for $a$. If $w_{a+1} \notin A(x)$, then we're done. Otherwise, since $y^a \ge x$, we have $y^a_{ \one( x ) } = 1$, therefore $\Dominating( y^a )_{w_{a+1}} = 0$ and hence $y^{a+1}_{A(x)} = 0$.

We now prove that if $x \mapsto_\Dominating y$, then $G[ \one( y_{B( x )} ) ]$ is acyclic. 
Suppose that $G[ \one( y_{B( x )} ) ]$ has a cycle $v_1, \dots, v_k$ and without loss, suppose that $v_k$ is the last updated. Let $z$ be the configuration before that last update, then $z_{v_{k-1}} = 1$ hence $y_{v_k} = \Dominating( z )_{v_k} = 0$, which is the desired contradiction.

We finally prove that if $y \ge x$, $y_{A( x )} = 0$, and $H = G[ \one( y_{B( x )} ) ]$ is acyclic, then there is a geodesic from $x$ to $y$. 
Let $w$ be a traversal of $H$ in reverse topological order; it is easily shown by induction on $a$ that $\InNeighbourhood( w_{a+1} ; G ) \subseteq \zero( y^a )$ and hence $\Dominating( y^a )_{ w_{a+1} } = 1$, as required.
\end{proof}

We now determine the words that fix the dominating network.

We denote the set of closed twins of $v$ by $\Twins{v}$ and we note that $\Twins{v} \subseteq \InNeighbourhood( v )$. Note that $V$ can be efficiently partitioned into these equivalence classes. 

\begin{proposition}[Words fixing the dominating network] \label{proposition:fixing_words_D}
Let $G$ be a digraph. Then $w$ fixes the dominating set network $\Dominating(G)$ if and only if $[w] \cap \Twins{m} \ne \emptyset$ for all $m \in \Benjamins(G)$.
\end{proposition}

\begin{proof}
Suppose that there exists $m \in \Benjamins(G)$ such that $[w] \cap \Twins{m} = \emptyset$, then we exhibit a configuration $x$ such that $y_{ \InNeighbourhood[m] } = 0$. 
Let $A = \InNeighbourhood[ m ] \setminus \Twins{m}$ and $B = \InNeighbourhood( A ) \setminus \InNeighbourhood[ m ]$. For any $a \in A$, there exists $b \in \InNeighbourhood( a ) \setminus \InNeighbourhood[ m ] \subseteq B$, since otherwise we would have $\InNeighbourhood[ a ] \subseteq \InNeighbourhood[ m ]$ and hence $a \in \Twins{m}$. Let $x_{ \Twins{m}, A, B } = (0, 0, 1)$, then $y_{ \Twins{m} } = x_{ \Twins{m} } = 0$, $y_B \ge x_B = 1$ and hence $y_A = 0$, thus $y_{\InNeighbourhood[ m ]} = 0$.

Suppose that for all $m \in \Benjamins(G)$, we have $[w] \cap \Twins{m} \ne \emptyset$, and suppose that $y_{ \InNeighbourhood[ v ] } = 0$ for some $v$. For all $u$ with $\InNeighbourhood[ u ] \subseteq \InNeighbourhood[ v ]$, we also have $y_{ \InNeighbourhood[ u ] } = 0$, therefore, there exists $m \in \Benjamins(G)$, $\InNeighbourhood[ m ] \subseteq \InNeighbourhood[ v ]$ such that $y_{ \InNeighbourhood[ m ] } = 0$. Since $y \ge x$, we also have $x_{ \InNeighbourhood[ m ]} = 0$. Let $m' \in [w] \cap \Twins{m} \subseteq \InNeighbourhood[ m ]$ and let $w_{a+1}$ be the first update of $m'$, then $y^a_{\InNeighbourhood[ m' ]} = y^a_{ \InNeighbourhood[ m ] } = 0$, thus $\Dominating( y^a )_{ w_{a+1} } = 1 = y_{m'}$ and $y_{ \InNeighbourhood[ m ] } \ne 0$, which is the desired contradiction.
\end{proof}

\begin{decisionproblem}
  \problemtitle{$\Dominating$ fixing word}
  \probleminput{A digraph $G = (V,E)$ and a word $w$.}
  \problemquestion{Does $w$ fix $\Dominating(G)$?}
\end{decisionproblem}

\begin{corollary} \label{corollary:complexity_fixing_word_D}
\DecisionProblem{$\Dominating$ fixing word} is in \P.
\end{corollary}

\section{Conclusion} \label{section:conclusion}

\paragraph{Summary of results}

In this paper, we have considered the generalisation of the MIS algorithm to allow for any initial configuration and to use update words that are not necessarily permutations. We have defined many decision problems with respect to this generalisation, such as: 
\begin{itemize}
    \item Given $G$ and a configuration $x$, can $x$ reach all MIS?

    \item Given $G$ and a word $w$, does $w$ fix $\MIS(G)$?

    \item Given $G$ and a set of vertices $S$, can we fix $\MIS(G)$ by only updating $S$?

    \item Given $G$, is there a word fixing $\MIS(G)$ that skips a vertex?

    \item Given $G$ and a permutation $w$, does $w$ fix $\MIS(G)$ (i.e. is $w$ a permis of $G$)?

    \item Given $G$, does $G$ have a permis?
\end{itemize} 
Even though every graph has a fixing word that guarantees terminating at a MIS regardless of the initial configuration, all the decision problems about the MIS algorithm in this paper are computationally hard. Additionally, we exhibit broad classes of graphs with and without permises, and relate these to existing graph classes. We introduce the class of near-comparability (a strict superclass of comparability graphs, which themselves encompass interval graphs and bipartite graphs, among others) and show that all near-comparability graphs are permissible. 

We further extended the MIS algorithm to digraphs; in this case, deciding whether the kernel network has a fixing word is computationally hard once again. Lastly, we consider the independent network and the dominating network, which are both related to the kernel network, and show that the analogous problems for these networks are tractable.

\paragraph{Future work} This paper can be extended in several ways. We give three potential avenues below.
\begin{enumerate}
    \item Graph classes. \\
    Since our problems are \NP- or \coNP-hard for the class of all graphs, it is natural to examine the complexity of those problems when we restrict ourselves to particular graph classes. The main tool for reductions is the \DecisionProblem{Constituency} decision problem. However, the reductions used in this paper did not preserve certain graph classes. For instance, \DecisionProblem{Constituency} remains \NP-complete even for bipartite graphs, while \DecisionProblem{Permissible} is trivial for comparability graphs.

    \item Minimum length of a fixing word. \\
    We have investigated the existence of fixing words, but not their lengths. From our results on prefixing and suffixing words, we get an upper bound on the minimum length of a fixing word: $a + b - 1$, where $a$ is the minimum size of a non-district and $b$ is the minimum size of a vertex cover. We conjecture that the problem of determining the minimum length of a fixing word is computationally hard. The analogous problem when we can only start at the all-zero configuration is obviously \NP-hard, as this amounts to determining the minimum size of a maximal independent set.

    \item Permises with bounded diameter. \\
    Let $w$ be a permutation of $V$, and for any vertex $v$, let $\delta(v)$ denote the maximum length of a path terminating at $v$ in $G^w$. Let $C_i = \{ v \in V : \delta( v ) = i \}$, then $V = C_0 \cup \dots \cup C_d$, where $d$ is the diameter of $w$. Instead of updating the vertices sequentially according to $w$, one can update all the vertices in $C_i$ at once, thus only requiring $d+1$ time steps. As such, the diameter of a permis $w$ measures the time it takes to fix the MIS network when we allow for some amount of synchronicity. If $\Delta$ is the maximum degree of $G$, then $G$ always has a permutation of diameter $\Delta$: partition $V$ into colour classes $C_0, \dots, C_\Delta$ (since the chromatic number is at most $\Delta + 1$), then $C_0 \succ C_1 \succ \dots \succ C_\Delta$. This is best possible if $G$ is complete, for instance. We therefore ask: if $G$ has a permis, then does it have a permis of diameter bounded by a function of $\Delta$?
\end{enumerate}

\section*{Acknowledgments}

We would like to thank William K. Moses Jr. and Peter Davies-Peck for their insights on self-stabilization and distributed computing.


\begin{thebibliography}{10}

\bibitem{AABCHK12}
Yehuda Afek, Noga Alon, Ziv Bar-Joseph, Alejandro Cornejo, Bernhard Haeupler, and Fabian Kuhn.
\newblock Beeping a maximal independent set.
\newblock {\em Distributed computing}, 26(4):195--208, 2013.

\bibitem{AABHBB11}
Yehuda Afek, Noga Alon, Omer Barad, Eran Hornstein, Naama Barkai, and Ziv Bar-Joseph.
\newblock A biological solution to a fundamental distributed computing problem.
\newblock {\em Science}, 331(6014):183--185, 2011.

\bibitem{ARS14}
J.~Aracena, A.~Richard, and L.~Salinas.
\newblock Maximum number of fixed points in and-or-not networks.
\newblock {\em Journal of Computer and System Sciences}, 80(7):1175 -- 1190, 2014.

\bibitem{ARS23}
J.~Aracena, A.~Richard, and L.~Salinas.
\newblock Synchronizing boolean networks asynchronously.
\newblock {\em Journal of Computer and System Sciences}, 136:249--279, 2023.

\bibitem{AGRS20}
Julio Aracena, Maximilien Gadouleau, Adrien Richard, and Lilian Salinas.
\newblock Fixing monotone boolean networks asynchronously.
\newblock {\em Information and Computation}, 274(104540), October 2020.

\bibitem{AGMS09}
Julio Aracena, Eric Goles, A.~Moreira, and Lilian Salinas.
\newblock On the robustness of update schedules in boolean networks.
\newblock {\em BioSystems}, 97:1--8, 2009.

\bibitem{BG09a}
Jorgen Bang-Jensen and Gregory Gutin.
\newblock {\em Digraphs: Theory, Algorithms and Applications}.
\newblock Springer, 2009.

\bibitem{BBDK18}
Joffroy Beauquier, Janna Burman, Fabien Dufoulon, and Shay Kutten.
\newblock {Fast Beeping Protocols for Deterministic MIS and ($\Delta$ + 1)-Coloring in Sparse Graphs}.
\newblock In {\em IEEE INFOCOM 2018 - IEEE Conference on Computer Communications}, pages 1754--1762, 2018.

\bibitem{BFS12}
Guy~E. Blelloch, Jeremy~T. Fineman, and Julian Shun.
\newblock Greedy sequential maximal independent set and matching are parallel on average.
\newblock In {\em SPAA '12: Proceedings of the twenty-fourth annual ACM symposium on Parallelism in algorithms and architectures}, pages 308--317, June 2012.

\bibitem{BGS93}
B{\'e}la Bollob{\'a}s and Imre Leader.
\newblock Connectivity and dynamics for random subgraphs of the directed cube.
\newblock {\em Israel Journal of Mathematics}, 83:321--328, 1993.

\bibitem{Bor08}
S.~Bornholdt.
\newblock Boolean network models of cellular regulation: prospects and limitations.
\newblock {\em Journal of The Royal Society Interface}, 5(Suppl 1):S85--S94, 2008.

\bibitem{CMRZ19}
A.~Casteigts, Y.~Métivier, J.M. Robson, and A.~Zemmari.
\newblock Design patterns in beeping algorithms: Examples, emulation, and analysis.
\newblock {\em Information and Computation}, 264:32--51, 2019.

\bibitem{CD19}
Artur Czumaj and Peter Davies.
\newblock Communicating with beeps.
\newblock {\em J. Parallel Distrib. Comput.}, 130(C):98–109, aug 2019.

\bibitem{Dij74}
Edsger~W. Dijkstra.
\newblock Self-stabilizing systems in spite of distributed control.
\newblock {\em Commun. ACM}, 17(11):643–644, nov 1974.

\bibitem{GR18}
Maximilien Gadouleau and Adrien Richard.
\newblock On fixable families of boolean networks.
\newblock In {\em Proc. Workshop on Asynchronous Cellular Automata}, pages 396--405, September 2018.

\bibitem{Gha19}
Mohsen Ghaffari.
\newblock Distributed maximal independent set using small messages.
\newblock In {\em Proceedings of the 2019 Annual ACM-SIAM Symposium on Discrete Algorithms (SODA)}, pages 805--820, 2019.

\bibitem{Gha22}
Mohsen Ghaffari.
\newblock Local computation of maximal independent set.
\newblock In {\em 2022 IEEE 63rd Annual Symposium on Foundations of Computer Science (FOCS)}, pages 438--449, 2022.

\bibitem{GM12}
E.~Goles and M.~Noual.
\newblock Disjunctive networks and update schedules.
\newblock {\em Advances in Applied Mathematics}, 48(5):646--662, 2012.

\bibitem{Gol85}
Eric Goles.
\newblock Dynamics of positive automata networks.
\newblock {\em Theoretical Computer Science}, 41:19--32, 1985.

\bibitem{Kar72}
Richard~M. Karp.
\newblock Reducibility among combinatorial problems.
\newblock In Raymond~E. Miller and James~W. Thatcher, editors, {\em Complexity of Computer Computations}, pages 85--103. Plenum Press, New York, 1972.

\bibitem{Kau69}
S.~A. Kauffman.
\newblock Metabolic stability and epigenesis in randomly connected nets.
\newblock {\em Journal of Theoretical Biology}, 22:437--467, 1969.

\bibitem{LSW09}
Christoph Lenzen, Jukka Suomela, and Roger Wattenhofer.
\newblock {Local Algorithms: Self-Stabilization on Speed}.
\newblock In S\'{a}ndor Fekete, Stefan Fischer, Martin Riedmiller, and Suri Subhash, editors, {\em Algorithmic Methods for Distributed Cooperative Systems}, volume 9371 of {\em Dagstuhl Seminar Proceedings (DagSemProc)}, pages 1--18, Dagstuhl, Germany, 2010. Schloss Dagstuhl -- Leibniz-Zentrum f{\"u}r Informatik.

\bibitem{MS97}
Ross~M. McConnell and Jeremy Spinrad.
\newblock Linear-time transitive orientation.
\newblock In {\em Proc. 8th ACM-SIAM Symposium on Discrete Algorithms}, pages 19--25, 1997.

\bibitem{NS17}
Mathilde Noual and Sylvain Sen{\'e}.
\newblock Synchronism versus asynchronism in monotonic boolean automata networks.
\newblock {\em Natural Computing}, 17:393--402, June 2018.

\bibitem{RRM13}
Landon Rabern, Brian Rabern, and Matthew Macauley.
\newblock Dangerous reference graphs and semantic paradoxes.
\newblock {\em Journal of Philosophical Logic}, 42(5):727--765, 2013.

\bibitem{RR13}
A.~Richard and P.~Ruet.
\newblock From kernels in directed graphs to fixed points and negative cycles in boolean networks.
\newblock {\em Discrete Applied Mathematics}, 161(7):1106--1117, 2013.

\bibitem{Rob80}
F.~Robert.
\newblock Iterations sur des ensembles finis et automates cellulaires contractants.
\newblock {\em Linear Algebra and its Applications}, 29:393--412, 1980.

\bibitem{Sch93}
Jukka Suomela.
\newblock Survey of local algorithms.
\newblock {\em ACM Comput. Surv.}, 45(2), mar 2013.

\bibitem{Tho73}
R.~Thomas.
\newblock {B}oolean formalization of genetic control circuits.
\newblock {\em Journal of Theoretical Biology}, 42(3):563 -- 585, 1973.

\bibitem{Yab93}
Steven Yablo.
\newblock Paradox without self-reference.
\newblock {\em Analysis}, 53(4):251--252, 1993.

\end{thebibliography}

\end{document}